\documentclass[final]{ws-idaqp}
\usepackage{graphicx}
\usepackage{graphicx}

\usepackage{amsmath}
\usepackage{amssymb}
\usepackage{amsfonts}
\usepackage{graphicx}
\usepackage{color}
\usepackage{times}
\usepackage{mathptmx} % assumes new font selection scheme installed
\usepackage{datetime}
\usepackage{hyperref}

\usepackage{framed} % for pdf, bitmapped graphics files
\usepackage{graphicx} % for pdf, bitmapped graphics files
\usepackage{amsmath} % assumes amsmath package installed
\usepackage{amssymb}  % assumes amsmath package installed
\newcommand{\sinc}{\mathrm{sinc}}
\newcommand{\tanc}{\mathrm{tanc}}
\newcommand{\tanhc}{\mathrm{tanhc}}
\def\supp{\mathrm{supp}}
\def\rank{\mathrm{rank}}
\DeclareMathAlphabet{\bit}{OML}{cmm}{b}{it}

\def\<{\leqslant}           % nice less than or equal to sign
\def\>{\geqslant}           % nice larger than or equal to sign

\def\d{\partial}

\def\Re{\mathrm{Re}}   % real part
\def\Im{\mathrm{Im}}   % imaginary part
\def\rprod{\mathop{\overrightarrow{\prod}}}

\def\cH{\mathcal{H}}   %
\def\mZ{\mathbb{Z}}    % set of integers
\def\mN{\mathbb{N}}    % set of positive integers
\def\mR{\mathbb{R}}    % real line
\def\mC{\mathbb{C}}    % complex plane

\def\Tr{\mathrm{Tr}}       % matrix trace
\def\rT{\mathrm{T}}        % matrix transpose
\def\rF{\mathrm{F}}        % Frobenius
\def\rHS{\mathrm{HS}}      % Hilbert-Schmidt
\def\bP{\mathbf{P}}    % probability
\def\bE{\mathbf{E}}    % quantum expectation
\def\bM{\mathbf{M}}    % classical expectation

\def\bra{{\langle}}
\def\ket{{\rangle}}

\def\re{\mathrm{e}}        % number e
\def\rd{\mathrm{d}}        % differential

\def\cL{\mathcal{L}}

\def\bD{\mathbf{D}}
\def\bJ{\mathbf{J}}

\def\br{\mathbf{r}}
\def\x{\times}
\def\ox{\otimes}

\def\fF{\mathfrak{F}}
\def\fH{\mathfrak{H}}
\def\fS{\mathfrak{S}}
\def\fR{\mathfrak{R}}
\def\fP{\mathfrak{P}}

\def\bH{\mathbf{H}}

\def\cF{\mathcal{F}}

\def\cX{\mathcal{X}}

\def\cK{\mathcal{K}}

\def\cC{\mathcal{ C}}

\def\cG{\mathcal{G}}
\def\cI{\mathcal{I}}
\def\cP{\mathcal{P}}
\def\cA{\mathcal{A}}
\def\cB{\mathcal{B}}
\def\cov{\mathbf{cov}}
\def\bG{\mathbf{G}}
\def\eps{\epsilon}
\def\ups{\upsilon}
\def\Ups{\Upsilon}

\begin{document}
\markboth{I.G.Vladimirov, I.R.Petersen \& M.R.James}
{Quadratic-exponential functionals of Gaussian
 quantum processes}

%%%%%%%%%%%%%%%%%%% Publisher's Area please ignore %%%%%%%%%%%%%%%%%%%%%%
\catchline{}{}{}{}{}
%%%%%%%%%%%%%%%%%%%%%%%%%%%%%%%%%%%%%%%%%%%%%%%%%%%%%%%%%%%%%%%%%%%%%%%%%

\title{\Large Quadratic-exponential functionals of Gaussian
 quantum processes}

\author{Igor G. Vladimirov\footnote{Corresponding author}, \quad
Ian R. Petersen$^\dagger$ and Matthew R. James$^\ddagger$
}

\address{
School of Engineering,
College of Engineering and Computer Science,
Australian National University,\\ Canberra,
Acton, ACT 2601, Australia\\
$^*$igor.g.vladimirov@gmail.com,
$^\dagger$i.r.petersen@gmail.com,
$^\ddagger$matthew.james@anu.edu.au
}

\maketitle

%\begin{history}
%\received{(Day Month Year)}
%\revised{(Day Month Year)}
%\published{(Day Month Year)}
%\comby{(xxxxxxxxxx)}
%\end{history}
\thispagestyle{empty}
\begin{abstract}
This paper is concerned with exponential moments of integral-of-quadratic functions of quantum processes
%formed by time-varying self-adjoint quantum variables
with canonical commutation relations of position-momentum type.  Such quadratic-exponential functionals (QEFs) arise as robust performance criteria in control problems for open quantum harmonic oscillators (OQHOs)   driven by bosonic fields.  We develop a randomised representation for the QEF using a Karhunen-Loeve  expansion of the quantum process on  a bounded time interval over the eigenbasis of its two-point commutator kernel, with noncommuting position-momentum pairs as coefficients.  This representation holds regardless of a particular quantum state and employs averaging over an auxiliary  classical Gaussian random process whose covariance operator is specified by the commutator kernel. This allows the QEF to be related to the moment-generating functional of the quantum process and computed for  multipoint Gaussian states. For stationary Gaussian quantum processes,
%such as those generated by the system variables of  a stable OQHO in its invariant state,
we establish a frequency-domain formula for the QEF rate in terms of the Fourier transform
%s of the real and imaginary parts
of the quantum covariance kernel in composition with trigonometric functions. A differential equation  is obtained for the QEF rate with respect to the risk sensitivity parameter for its approximation  and numerical computation.
% via a homotopy technique.
The QEF is also applied to large deviations and worst-case mean square cost bounds for OQHOs in the presence of statistical uncertainty with a quantum relative entropy description.
\end{abstract}


\ccode{AMS Subject Classification:
81S22, % Open systems, reduced dynamics, master equations, decoherence
81S25, % Quantum stochastic calculus
81P16, % Quantum state spaces, operational and probabilistic concepts
81S05,      %Canonical quantization, commutation relations and statistics
81R15,  % 	Operator algebra methods applied to problems in quantum theory
47B35,  % Toeplitz operators, Hankel operators, Wiener-Hopf operators
81Q10,   	%Selfadjoint operator theory in quantum theory, including spectral analysis
37L40,      % Invariant measures
60G15,   	% Gaussian processes
81Q93,   	%Quantum control
93B35,      %Sensitivity (robustness)
94A17.       % Measures of information, entropy
}

\section{Introduction}

Interaction of quantum mechanical systems with their environment (involving, for example,  other quantum systems, external quantum fields or classical measuring devices), which is the main subject of open quantum dynamics, is viewed in quantum control also from the point of achieving certain dynamic properties for such systems. This gives rise to control and filtering settings for quantum networks, which can consist of a quantum plant and a quantum or classical feedback controller or observer with direct or field-mediated coupling, or more complicated interconnections of several quantum systems\cite{GJ_2009,JG_2008}.
Tractable models for such systems with continuous variables include, in particular, open quantum harmonic oscillators (OQHOs) governed by linear quantum stochastic differential equations (QSDEs)  in the framework of the Hudson-Parthasarathy calculus\cite{HP_1984,P_1992,P_2015}. These models are used in linear quantum systems theory\cite{NY_2017,P_2017} in combination with performance criteria, which are organised as cost functionals to be minimised over admissible system interconnections and their parameters.

Similarly to the use of quadratic cost functionals (such as the mean square value of the estimation error) in the Kolmogorov-Wiener-Hopf-Kalman filtering and linear-quadratic-Gaussian control theories\cite{AM_1971} for classical linear stochastic systems, linear quantum control also employs quadratic performance criteria\cite{EB_2005,NJP_2009,WM_2010}.  In particular, quadratic cost functionals and their minimization provide a natural way to quantify and improve the performance of observers in quantum filtering problems in terms of the mean square discrepancy between the system variables and their estimates\cite{MJ_2012}. The mean square optimality criteria for linear quantum stochastic systems are complemented by (and are a limiting case of)  the quantum mechanical adaptation of quadratic-exponential cost functionals,  which originate  from classical risk-sensitive control\cite{BV_1985,J_1973,W_1981}.

The quadratic exponential functional\cite{VPJ_2018a} (QEF) (see also Ref.~\refcite{B_1996}) retains the general structure of  its classical counterparts and is organised as the averaged exponential of the integral of a  quadratic function of the system variables over a bounded time interval. Being a higher-order mixed moment of the quantum variables, the QEF leads to exponential upper bounds on the tail distributions for quadratic functions of the quantum system trajectories\cite{VPJ_2018a}, which corresponds to the large deviations theory for classical  random processes\cite{DE_1997,V_2008}. Another useful property of the QEF is its relation\cite{VPJ_2018b} with upper bounds on the worst-case values of mean square costs in the presence of quantum statistical uncertainty,  when the actual system-field state differs from its nominal model, but the deviation is limited in terms of quantum relative entropy\cite{OP_1993,OW_2010,YB_2009}.  These properties  involve the QEF in such a way that its minimization makes the dynamic behaviour of the open quantum system more conservative and robust. The robustness  considerations are relevant, for example,  for applications to quantum optics\cite{WM_2008} and quantum information processing\cite{NC_2000}, which demand certain insensitivity to unmodelled nonlinear dynamics along with the possibility to isolate the quantum  system from its surroundings in a controlled fashion.

The resulting robust performance analysis and optimal control problems require methods for computing and  minimizing the QEF, which is different both from its classical predecessors and the quantum risk-sensitive control formulation\cite{J_2004,J_2005} with time-ordered exponentials despite the Lie-algebraic links\cite{VPJ_2019a} between these two classes of cost functionals.  These differences  are caused by the noncommutativity  of the underlying quantum variables and have close algebraic and quantum probabilistic connections with the operator exponential structures arising in the context
of operator algebras\cite{AB_2018}, quantum mechanical extensions of the  L\'{e}vy area\cite{CH_2013,H_2018},
and the moment-generating and partition functions for quadratic Hamiltonians in quantum statistical mechanics\cite{BB_2010,PS_2015,S_1994}.

The development of methods for computing the QEF has been a subject of several recent publications based on a parametric randomization technique of Ref.~\refcite{VPJ_2018c}. This result allows the exponential moment of a quadratic function of a finite number of quantum variables with canonical commutation relations (CCRs), similar to those of the positions and momenta, to be represented in a randomised form. This representation involves classical averaging over auxiliary independent standard normal random variables as parameters of the moment-generating function of the quantum variables. The randomised representation was then extended to continuous time using quantum Karhunen-Loeve (QKL) expansions\cite{VPJ_2019b,VJP_2019}  for system variables of OQHOs over different basis functions, with the QKL expansion being a quantum counterpart of its predecessor for classical random processes\cite{GS_2004,IT_2010}. These results have led to an integral operator representation\cite{VPJ_2020a} of the QEF over bounded time intervals  for OQHOs in multipoint Gaussian quantum states\cite{KRP_2010,VPJ_2018a}. This finite-horizon representation has subsequently  been used in Ref.~\refcite{VPJ_2020b} in order to establish a frequency-domain formula for the infinite-horizon  asymptotic growth rate of the logarithm of the QEF for invariant  Gaussian states of stable OQHOs, driven by vacuum bosonic fields\cite{P_1992}.

The purpose of the present paper is a systematic extension of the results, briefly announced in the  above-mentioned conference papers\cite{VPJ_2018b,VPJ_2019b,VJP_2019,VPJ_2020a,VPJ_2020b}, to a wider class of quantum processes, which contain,  as a particular case, those formed from the system variables of OQHOs.
More precisely, the paper is concerned with exponential moments of integral-of-quadratic functions of quantum processes formed by time-varying self-adjoint quantum variables satisfying CCRs similar to those of the positions and momenta.
We develop a Girsanov type representation for the QEF using a QKL expansion of the quantum process on  a bounded time interval over the real and imaginary parts of the orthonormal eigenfunctions of a skew self-adjoint integral operator with the two-point commutator kernel of the process, provided the operator has no zero eigenvalues.  The coefficients of the QKL expansion are organised as noncommuting position-momentum pairs.  The resulting representation of the QEF
is valid regardless of a particular quantum state and employs averaging over an auxiliary  classical Gaussian random process whose covariance operator is defined in terms of the quantum commutator kernel. We use this representation in order to relate the QEF to the moment-generating functional (MGF) of the quantum process and compute it when this process is in a multipoint Gaussian state.
For stationary zero-mean  Gaussian quantum processes (including but not limited to those generated by the system variables of  a stable OQHO in its invariant state), we obtain a frequency-domain formula for the QEF growth rate in terms of the Fourier transforms of the real and imaginary parts of the quantum covariance function in composition with trigonometric functions. This leads to a differential equation for the QEF growth rate with respect to the risk sensitivity parameter, which is applicable to its asymptotic approximation and numerical computation using a homotopy method similar to that for solving parameter dependent algebraic equations\cite{MB_1985}. These results are also specified for  invariant Gaussian states of stable OQHOs driven by vacuum fields, in which case the quantum covariance  function of the stationary Gaussian quantum process, linearly related to the system variables, has a rational Fourier transform. In this case, the eigenvalue problem for the commutator kernel is related to a boundary value problem for a second-order ODE, and a sufficient condition is obtained for the absence of zero eigenvalues. For stable OQHOs, we  also apply the QEF growth rate to upper bounds on the worst-case growth rate of a mean square cost in the presence of statistical uncertainty described in terms of the quantum relative entropy rate of the actual system-field state with respect to its nominal model.

The paper is organised as follows.
Section~\ref{sec:proc} specifies the class of quantum processes with position-momentum type CCRs being considered.
Section~\ref{sec:eig} discusses the eigenvalues and eigenfunctions of the skew self-adjoint integral operator with the two-point commutator kernel of the process.
Section~\ref{sec:QKL} employs this eigenbasis for a finite-horizon QKL expansion of the process.
Section~\ref{sec:quadro} specifies a QEF for the quantum process
and applies the QKL expansion  in order to represent this functional in terms of averaging over an auxiliary  classical random sequence.
Section~\ref{sec:classproc} develops a continuous-time version of this randomised representation using a classical  Gaussian process and the moment-generating functional of the quantum process.
Section~\ref{sec:gauss} specifies this representation for multi-point Gaussian quantum states.
Section~\ref{sec:stat}  establishes the infinite-horizon asymptotic growth rate for the logarithm of the QEF for stationary zero-mean Gaussian quantum processes.
Section~\ref{sec:comp} discusses a homotopy technique for computing the QEF growth rate and its asymptotic expansion over the risk sensitivity parameter.
Section~\ref{sec:sys} specifies  the above results to the system variables of a stable OQHO driven by vacuum input fields.
Section~\ref{sec:worst} applies the QEF rate to upper bounds on the worst-case quadratic costs for OQHOs in the presence of statistical uncertainty.
Section~\ref{sec:conc} makes concluding remarks.  \ref{sec:A} provides a randomised representation for elementary quadratic-exponential functions of position-momentum pairs, which is used in Section~\ref{sec:quadro}.

\section{Quantum Processes Being Considered}
\label{sec:proc}

We consider a quantum process $X:= (X_k)_{1\< k \< n}$ consisting of self-adjoint operators $X_1(t), \ldots, X_n(t)$ on a complex separable Hilbert space $\fH$. It is assumed that these quantum variables are weakly continuous functions\cite{Y_1980}   of time $t\> 0$ with the two-point canonical commutation relations (CCRs)
\begin{equation}
\label{XXcomm}
    [X(s), X(t)^\rT]
    :=
    ([X_j(s),X_k(t)])_{1\< j,k\< n}
    =
    2i\Lambda(s,t) \ox \cI_\fH,
    \qquad
    s,t\>0
\end{equation}
on a common dense domain in $\fH$, specified by a continuous function  $\Lambda: \mR_+^2\to \mR^{n\x n}$ (with $\mR_+:= [0,+\infty)$) satisfying
\begin{equation}
\label{skew}
    \Lambda(s,t) = -\Lambda(t,s)^\rT,
    \qquad
    s,t\> 0.
\end{equation}
Here,   $(\cdot)^\rT$ is the matrix transpose (vectors are organised as columns unless indicated otherwise),  $[\alpha,\beta]:= \alpha\beta - \beta\alpha$ is the commutator of linear operators, $i:= \sqrt{-1}$ is the imaginary unit, $\ox$ is the tensor product of spaces or operators (in particular, the Kronecker product of matrices), and  $\cI_\fH$ is the identity operator on $\fH$. For any complex matrix $C \in \mC^{a\x b}$, the matrix $C\ox\cI_\fH$ is identified with $C$.

In view of (\ref{XXcomm}), (\ref{skew}),  the real antisymmetric matrix $\Lambda(t,t)$ of order $n$ describes the one-point CCRs for the process $X$:
\begin{equation}
\label{Xcomm}
    [X(t), X(t)^\rT]
    =
    2i\Lambda(t,t),
    \qquad
    t\>0.
\end{equation}
Such CCRs hold for pairs
of the quantum mechanical position and momentum  operators\cite{S_1994}   $\xi$ and $\eta:= -i\d_\xi$,  implemented on the Schwartz space\cite{V_2002} and satisfying
$[\xi,\eta] = i$, so that the vector
\begin{equation}
\label{zeta}
    \zeta:=
    {\begin{bmatrix}
        \xi\\
        \eta
    \end{bmatrix}}
\end{equation}
has the CCR matrix $\frac{1}{2}\bJ$ in the sense that
\begin{equation}
\label{posmomCCR}
    [\zeta, \zeta^\rT]
    =
    \begin{bmatrix}
        [\xi,\xi] & [\xi,\eta]\\
        [\eta,\xi] & [\eta,\eta]
    \end{bmatrix}
    =
    i\bJ,
\end{equation}
where
\begin{equation}
\label{bJ}
\bJ: = {\begin{bmatrix}
        0 & 1\\
        -1 & 0
    \end{bmatrix}}
\end{equation}
spans the subspace of antisymmetric matrices of order 2. The conjugate position-momentum pairs provide building blocks for more complicated CCRs between several (or an infinite number of) quantum variables.

For a fixed but otherwise arbitrary time horizon $T>0$,   the two-point CCR function $\Lambda$ in (\ref{XXcomm})  gives rise to an integral operator $\cL$,  which acts on a function
$f \in L^2([0,T],\mC^n)$
as
\begin{equation}
\label{cL}
    \cL(f)(s)
    :=
    \int_0^T
    \Lambda(s,t) f(t)
    \rd t,
    \qquad
    0\< s \< T,
\end{equation}
and is skew self-adjoint in view of (\ref{skew}).
The Hilbert spaces $L^2([0,T],\mC^r)$ of square integrable functions $f,g: [0,T]\to \mC^r$ are endowed with the inner product
$
    \bra f,g\ket := \int_0^T f(t)^* g(t)\rd t
$ and the norm $\|f\|:= \sqrt{\bra f,f\ket} = \sqrt{\int_0^T |f(t)|^2\rd t}$, where $(\cdot)^*:= {{\overline{(\cdot)}}}^\rT$ is the complex conjugate transpose, and $|\cdot|$ is the standard Euclidean norm.
With its kernel function $\Lambda$  being continuous, the operator $\cL$ is compact\cite{RS_1980}.

\section{Eigenbasis for the Two-Point Commutator Kernel}
\label{sec:eig}

Since the operator $\cL$ in (\ref{cL}) is skew self-adjoint (and hence, $i\cL$ is self-adjoint), its eigenvalues are purely imaginary. Furthermore, they are symmetric about the origin. Indeed, since $\cL$ maps the real subspace $L^2([0,T], \mR^n)$ into itself, then for any eigenfunction $f: [0,T]\to \mC^n$  of this operator, with
\begin{equation}
\label{freim}
    \varphi:= \Re f,
    \qquad
    \psi := \Im f,
\end{equation}
and eigenvalue $i\omega$, so that $\cL(f) = i\omega f$, with $\omega\in \mR$, the function $\overline{f} = \varphi-i\psi $ satisfies $\cL(\overline{f}) = \cL(\varphi)-i\cL(\psi) = \overline{\cL(f)} = -i\omega \overline{f}$ and is, therefore,  an eigenfunction of $\cL$ with the eigenvalue $-i\omega$. This property is represented in vector-matrix form as
\begin{equation}
\label{Lff}
    \begin{bmatrix}
      \cL(f) &
      \cL(\overline{f})
    \end{bmatrix}
    =
    i
    \omega
    \begin{bmatrix}
      f &
      \overline{f}
    \end{bmatrix}
    \begin{bmatrix}
      1 & 0\\
      0 & -1
    \end{bmatrix},
\end{equation}
where both sides of the equality are $\mC^{n\x 2}$-valued functions on $[0,T]$.
For what follows, we assume that the null space of the operator $\cL$ is trivial:
\begin{equation}
\label{non0}
    \ker \cL = \{0\},
\end{equation}
that is, $\cL$
    has no zero eigenvalues.
This means that there does not exist a nonzero  $f \in L^2([0,T], \mC^n)$ which makes  $\big[X(s), \int_0^T f(t)^\rT X(t)\rd t\big] = 2i\cL(f)(s)$ vanish for all $0\< s \< T$. The property (\ref{non0}) can therefore  be interpreted as ``complete noncommutativity'' of the quantum process $X$ over the time interval $[0,T]$.

Let $f\in L^2([0,T],\mC^n)\setminus \{0\}$ be an eigenfunction of $\cL$ with the eigenvalue $i\omega$, where $\omega >0$ (without loss of generality under the condition (\ref{non0})) will be referred to as an eigenfrequency of $\cL$.  The functions
$\varphi, \psi \in L^2([0,T], \mR^n)$ in (\ref{freim}) satisfy
$\|\varphi\|^2 + \|\psi\|^2 = \|f\|^2 = \|\overline{f}\|^2$ and are related by
$\cL(\varphi) = -\omega \psi$ and $\cL(\psi) = \omega \varphi$, or, in vector-matrix form,
\begin{equation}
\label{Lphipsi}
    \begin{bmatrix}
      \cL(\varphi) &
      \cL(\psi)
    \end{bmatrix}
    =
    \omega
    \begin{bmatrix}
      \varphi &
      \psi
    \end{bmatrix}
        \bJ  ,
\end{equation}
where $\bJ$ is the matrix from (\ref{bJ}).  The equality (\ref{Lphipsi}) is a real-valued version of (\ref{Lff}) in view of the identities
\begin{equation}
\label{ffphipsi}
        {\begin{bmatrix}
      f &
      \overline{f}
    \end{bmatrix}}
    =
    \begin{bmatrix}
      \varphi &
      \psi
    \end{bmatrix}
    \Delta
\end{equation}
and
$
    i
    \Delta
    {\small\begin{bmatrix}
        1 &  0\\
         0 &  -1
     \end{bmatrix}}
     \Delta^{-1}
    =
    \bJ
$,
where
\begin{equation}
\label{Del}
    \Delta
    :=
        \begin{bmatrix}
      1 & 1\\
      i & -i
    \end{bmatrix},
\end{equation}
with the matrix $\frac{1}{\sqrt{2}} \Delta$ being unitary. Since the eigenfunctions $f$, $\overline{f}$ correspond to the different eigenvalues $\pm i\omega$ of the skew self-adjoint operator $\cL$, they are orthogonal, and their Gram matrix takes the form
$
    \bG(f,\overline{f})
    :=
    \int_0^T
        {\small\begin{bmatrix}
      f(t) &
      \overline{f(t)}
    \end{bmatrix}}^*
        {\small\begin{bmatrix}
      f(t) &
      \overline{f(t)}
    \end{bmatrix}}
    \rd t
    =
    {\small\begin{bmatrix}
      \|f\|^2 & \bra f, \overline{f}\ket \\
      \bra \overline{f}, f\ket & \|f\|^2
    \end{bmatrix}}
    =
    \|f\|^2 I_2
$,
where $I_r$ is the identity matrix of order $r$. On the other hand, it follows from (\ref{ffphipsi}) that
$
    \bG(f,\overline{f}) = \Delta^*\bG(\varphi,\psi)\Delta
$,
and hence, in view of the symmetry of $\bra\cdot, \cdot\ket$ on the real subspace $L^2([0,T],\mR^n)$ and the unitarity of the matrix $\frac{1}{\sqrt{2}}\Delta$,
$
    \bG(\varphi,\psi)
    =
    {\small\begin{bmatrix}
      \|\varphi\|^2 & \bra \varphi, \psi\ket \\
      \bra \varphi, \psi\ket  & \|\psi\|^2
    \end{bmatrix}}
    =
    \frac{1}{2}\Delta \bG(f,\overline{f}) \Delta^* = \frac{1}{2}\|f\|^2 I_2
$,
which is equivalent to
\begin{equation}
\label{norm1_norm2}
    \|\varphi\|^2 = \|\psi\|^2 = \frac{1}{2}\|f\|^2,
    \qquad
    \bra \varphi, \psi\ket = 0.
\end{equation}
If $f$, $g$ are eigenfunctions of $\cL$ with eigenvalues $i\omega$ and $i\mu$, so that $\cL(f) = i\omega f$ and $\cL(g) = i\mu g$, then $\omega\ne \mu$ implies $\bra f,g\ket = 0$. In particular,  if $\omega >0$ and $\mu>0$, then $f$, $\overline{g}$ correspond to the different eigenvalues $i\omega\ne -i\mu$, and hence, $\bra f, \overline{g}\ket = 0$. The eigenfunctions with a common eigenvalue are orthonormalised by the Gram-Schmidt procedure\cite{RS_1980}.

For what follows, all the eigenfrequencies $\omega_k>0$ of the operator $\cL$ in (\ref{cL})  are numbered by positive integers $k \in \mN:= \{1,2,3,\ldots\}$. The eigenfunctions
\begin{equation}
\label{eigff}
    f_k = \varphi_k + i\psi_k,
    \qquad
    \overline{f_k} = \varphi_k - i\psi_k
\end{equation}
of $\cL$, with the corresponding eigenvalues $\pm i \omega_k $,   and  their real and imaginary parts
\begin{equation}
\label{reim}
  \varphi_k := \Re f_k,
  \qquad
  \psi_k := \Im f_k
\end{equation}
satisfy
\begin{equation}
\label{ffff}
    \bra f_j,f_k\ket = \delta_{jk},
    \qquad
    \bra \overline{f_j},\overline{f_k}\ket = \overline{\bra f_j,f_k\ket} = \delta_{jk},
    \qquad
    \bra f_j, \overline{f_k}\ket = 0,
\end{equation}
or equivalently,
\begin{equation}
\label{reimort}
    \bra \varphi_j,\varphi_k\ket
    =
    \bra \psi_j,\psi_k\ket
    =
    \frac{1}{2}\delta_{jk},
    \qquad
    \bra \varphi_j,\psi_k\ket = 0,
    \qquad
    j,k\in \mN,
\end{equation}
with $\delta_{jk}$ the Kronecker delta, where the relations (\ref{norm1_norm2}) are also used.  An equivalent form of (\ref{reimort}) is
\begin{equation}
\label{hh}
  \int_0^T
  h_j(t)^\rT h_k(t)
  \rd t
  =
  {\begin{bmatrix}
    \bra \varphi_j, \varphi_k\ket & \bra \varphi_j, \psi_k\ket\\
    \bra \psi_j, \varphi_k\ket & \bra \psi_j, \psi_k\ket
  \end{bmatrix}}
  =
  \frac{1}{2}\delta_{jk}
  I_2,
\end{equation}
where the functions $h_k \in L^2([0,T], \mR^{n\x 2})$ are formed from the $\mR^n$-valued functions in (\ref{reim}) as
\begin{equation}
\label{hk}
    h_k:=
    {\begin{bmatrix}
      \varphi_k &
      \psi_k
    \end{bmatrix}}.
\end{equation}
The eigenfunctions (\ref{eigff}) can be represented as
$
        {\begin{bmatrix}
      f_k &
      \overline{f_k}
    \end{bmatrix}}
    =
    h_k
    \Delta
$
in terms of (\ref{hk}) by using (\ref{Del}). The resolution of the identity over the orthonormal eigenfunctions of the operator $\cL$ takes the form
$   f
   =
  \sum_{k=1}^{+\infty}
  (\bra f_k, f\ket f_k +\bra \overline{f_k},f\ket \overline{f_k})
$
for any   $  f \in L^2([0,T], \mC^n)$, or formally,
\begin{equation}
\label{cI}
    2
    \sum_{k=1}^{+\infty}
    \Re
    (f_k(s)f_k(t)^*) = \delta(s-t) I_n,
    \qquad
    0\< s,t \< T,
\end{equation}
where $\delta$ is the Dirac delta-function. Hence, the kernel function $\Lambda$ admits the following expansion over the eigenbasis:
\begin{align}
\nonumber
    \Lambda(s,t)
    & =
    i
    \sum_{k=1}^{+\infty}
    \omega_k
    (f_k(s)f_k(t)^*-\overline{f_k(s)}f_k(t)^\rT)
    = -2
    \sum_{k=1}^{+\infty}
    \omega_k
    \Im(f_k(s)f_k(t)^*),\\
\label{Lambdaseries}
    & =
    2
    \sum_{k=1}^{+\infty}
    \omega_k
    (\varphi_k(s)\psi_k(t)^\rT-\psi_k(s)\varphi_k(t)^\rT)
    =
    2
    \sum_{k=1}^{+\infty}
    \omega_k
    h_k(s) \bJ h_k(t)^\rT,
    \qquad
    0\< s,t\< T,
\end{align}
where the matrix $\bJ$ is given by (\ref{bJ}). This expansion
is similar to the Mercer representation\cite{KZPS_1976} for positive semi-definite self-adjoint integral operators with continuous kernels (see also Ref.~\refcite{BP_2012} and references therein).  Since $-i\cL$ is a self-adjoint operator on $L^2([0,T], \mC^n)$ with the kernel $-i\Lambda$, the eigenvalues $\pm \omega_k$ and the corresponding eigenfunctions $f_k$, $\overline{f_k}$, then, in accordance with (\ref{Lambdaseries}), (\ref{hh}), the squared  Hilbert-Schmidt norm\cite{RS_1980}  of the operator $\cL$ in (\ref{cL}) is
\begin{equation}
\label{HSnorm}
    \|\cL\|_{\rHS}^2
    =
    \int_{[0,T]^2}
    \|\Lambda(s,t)\|_\rF^2
    \rd s\rd t
    =
    2\sum_{k=1}^{+\infty}
    \omega_k^2
    =
    -
    \Tr (\cL^2) <+\infty,
\end{equation}
where $\|N\|_\rF := \sqrt{\Tr (N^*N)}$ is the Frobenius norm of a real or complex matrix $N$. The trace of the positive definite self-adjoint operator $-\cL^2$  in (\ref{HSnorm}) takes into account  the double multiplicity of the eigenvalues $-(\pm i\omega_k)^2 = \omega_k^2$.

If $\phi$ is a function of a complex variable, holomorphic in a neighbourhood of the line segment $D:= i\|\cL\|[-1, 1] \subset i\mR$ (containing the spectrum of the operator $\cL$), where $ \|\cL\| = \max_{k \> 1}\omega_k$ is the $L^2$-induced operator norm of $\cL$,
then $\phi(\cL)$
is a bounded operator on $L^2([0,T], \mC^n)$ with the eigenvalues $\phi(\pm i\omega_k)$ and the eigenfunctions (\ref{eigff}) inherited from $\cL$.
 In view of (\ref{cI}), (\ref{Lambdaseries}), the kernel function of $\phi(\cL)$ is given by
\begin{align}
\nonumber
    \sum_{k=1}^{+\infty}
    (\phi(i\omega_k)f_k(s)f_k(t)^*&+\phi(-i\omega_k)\overline{f_k(s)}f_k(t)^\rT)
     =
    2
    \sum_{k=1}^{+\infty}
    (\phi_+(i\omega_k)\Re(f_k(s)f_k(t)^*)
    +
    i\phi_-(i\omega_k)
    \Im(f_k(s)f_k(t)^*)),\\
\label{phiLambdaseries}
    & =
    2
    \sum_{k=1}^{+\infty}
    h_k(s)
    (\phi_+(i\omega_k)I_2 - i\phi_-(i\omega_k)\bJ)
    h_k(t)^\rT,
    \qquad
    0\< s,t\< T,
\end{align}
where $\phi_{\pm}(z):=\frac{1}{2}(\phi(z)\pm \phi(-z))$ are the symmetric and antisymmetric  parts of the function $\phi$, so that $\phi(\pm z) = \phi_+(z)\pm \phi_-(z)$. In particular, if $\phi$ is symmetric (and hence, $\phi'(0) = 0$) and satisfies $\phi(0)=0$,  then (\ref{HSnorm}) implies that  the operator $\phi(\cL)$ is of trace class\cite{RS_1980}, with
\begin{equation}
\label{phiLtrace}
    \Tr \phi(\cL)
    =
    2\sum_{k=1}^{+\infty} \phi(i\omega_k),
\end{equation}
where the series is absolutely convergent in view of the
asymptotic relation $\phi(i\omega) \sim -\frac{1}{2}\phi''(0)\omega^2 $, as $\omega\to 0$.

 \section{Quantum Karhunen-Loeve Expansion Using the Commutator Kernel Eigenbasis}
\label{sec:QKL}

With the quantum process $X$ over the time interval $[0,T]$, we
associate the following sequence of quantum variables on the Hilbert space $\fH$:
\begin{equation}
\label{gammak}
  \gamma_k
   :=
  \frac{1}{\sqrt{\omega_k}}
  \int_0^T
  f_k(t)^\rT
  X(t)
  \rd t
  =
  \xi_k + i\eta_k,
\end{equation}
where
\begin{equation}
\label{xik_etak}
    \xi_k
    :=
    \Re
    \gamma_k
    =
  \frac{1}{\sqrt{\omega_k}}
  \int_0^T
  \varphi_k(t)^\rT
  X(t)
  \rd t,
  \qquad
    \eta_k
    :=
    \Im
    \gamma_k
    =
  \frac{1}{\sqrt{\omega_k}}
  \int_0^T
  \psi_k(t)^\rT
  X(t)
  \rd t,
    \qquad
   k \in \mN,
\end{equation}
and the integrals are understood in the weak sense\cite{Y_1980}.  Here,
the orthonormal eigenfunctions (\ref{eigff}) of the operator $\cL$ in (\ref{cL}) with the two-point CCR kernel $\Lambda$ from (\ref{XXcomm}) are used together with (\ref{reim}) and  the eigenfrequencies $\omega_k$. Also, the real and imaginary parts are extended from $\mC$
to quantum variables as
$\Re z := \frac{1}{2} (z + z^\dagger)$,
$\Im z := \frac{1}{2i} (z - z^\dagger)$, where $(\cdot)^\dagger$ is the operator adjoint. Accordingly,
\begin{equation}
\label{gammak+}
  \gamma_k^\dagger
  =
  \frac{1}{\sqrt{\omega_k}}
  \int_0^T
  f_k(t)^*
  X(t)
  \rd t
  =
  \xi_k - i\eta_k,
\end{equation}
in view of (\ref{gammak}),
since the quantum variables $X_1(t), \ldots, X_n(t)$ and $\xi_k$, $\eta_k$ in (\ref{xik_etak}) are self-adjoint.

\begin{lemma}
\label{lem:posmom}
Suppose the operator $\cL$ in (\ref{cL}) with the continuous commutator kernel $\Lambda$ of the quantum process $X$ in (\ref{XXcomm}) over a given time horizon $T>0$ satisfies (\ref{non0}). Then the quantum variables $\gamma_k$, associated by (\ref{gammak}), (\ref{xik_etak}) with the process $X$ and the eigenbasis (\ref{eigff})--(\ref{reimort})   of the operator $\cL$, satisfy
\begin{equation}
\label{gcomm}
  [\gamma_j, \gamma_k^\dagger] = 2\delta_{jk},
  \quad
  [\gamma_j, \gamma_k] = 0,
  \quad
  [\gamma_j^\dagger, \gamma_k^\dagger] = 0,
          \qquad
        j,k\in \mN.
\end{equation}
\hfill$\square$
\end{lemma}
\begin{proof}
In view of (\ref{XXcomm}), (\ref{cL}), (\ref{ffff}), it follows from (\ref{gammak}), (\ref{gammak+}) that
$    [\gamma_j,\gamma_k^\dagger]
    =
  \frac{1}{\sqrt{\omega_j\omega_k}}
  \int_{[0,T]^2}
  f_j(s)^\rT
  [X(s),X(t)^\rT]
  \overline{f_k(t)}
  \rd s\rd t
  =
  \frac{2i}{\sqrt{\omega_j\omega_k}}
  \int_{[0,T]^2}
  f_j(s)^\rT
  \Lambda(s,t)
  \overline{f_k(t)}
  \rd s\rd t
  =
  \frac{2i}{\sqrt{\omega_j\omega_k}}
  \bra
    \overline{f_j},
    \cL(\overline{f_k})
  \ket
 =
  2
  \sqrt{\frac{\omega_k}{\omega_j}}
  \bra
    \overline{f_j},
    \overline{f_k}
  \ket
  =
        2\delta_{jk}
$ for all         $j,k\in \mN$, which establishes the first equality in (\ref{gcomm}).
By a similar reasoning,
$    [\gamma_j,\gamma_k]
    =
  \frac{1}{\sqrt{\omega_j\omega_k}}
  \int_{[0,T]^2}
  f_j(s)^\rT
  [X(s),X(t)^\rT]
  f_k(t)
  \rd s\rd t
  =
  \frac{2i}{\sqrt{\omega_j\omega_k}}
  \int_{[0,T]^2}
  f_j(s)^\rT
  \Lambda(s,t)
  f_k(t)
  \rd s\rd t
  =
  \frac{2i}{\sqrt{\omega_j\omega_k}}
  \bra
    \overline{f_j},
    \cL(f_k)
  \ket
 =
  -2
  \sqrt{\frac{\omega_k}{\omega_j}}
  \bra
    \overline{f_j},
    f_k
  \ket
  =
  0
$,
and hence,
$    [\gamma_j^\dagger,\gamma_k^\dagger]
    =
    -
    [\gamma_j,\gamma_k]^\dagger
    =0
$
for all $j,k\in \mN$,
which completes the proof of (\ref{gcomm}).
\end{proof}

The CCRs (\ref{gcomm}) show that $\gamma_k$ in (\ref{gammak}) are organised as pairwise commuting annihilation operators (with $\gamma_k^\dagger$ in (\ref{gammak+}) the corresponding  creation operators), so that the self-adjoint quantum variables $\xi_k$, $\eta_k$ in (\ref{xik_etak}) are conjugate pairs of the quantum mechanical  positions and momenta mentioned in Sec.~\ref{sec:proc}, with
\begin{equation}
\label{xietacomm}
  [\xi_j,\xi_k] = 0,
  \quad
  [\eta_j,\eta_k] = 0,
  \quad
  [\xi_j,\eta_k] = i \delta_{jk},
  \qquad
  j,k\in \mN.
\end{equation}
Accordingly, the vectors
\begin{equation}
\label{zetak}
  \zeta_k
  :=
  {\begin{bmatrix}
    \xi_k\\
    \eta_k
  \end{bmatrix}}
\end{equation}
commute between each other (for different $k$) and have a common CCR matrix $\frac{1}{2}\bJ$, since
$
    [\zeta_j,\zeta_k^\rT]
    =
    {\small\begin{bmatrix}
      [\xi_j,\xi_k] & [\xi_j,\eta_k]\\
      [\eta_j,\xi_k] & [\eta_j,\eta_k]
    \end{bmatrix}}
    =
    i \delta_{jk}
    \bJ
$,
where (\ref{xietacomm}) is used together with (\ref{bJ}). In view of the orthonormality of the  eigenbasis (\ref{eigff}), the process $X$ can be recovered from (\ref{gammak}), (\ref{gammak+}) as
\begin{align}
\nonumber
    X(t)
    &=
    \sum_{k=1}^{+\infty}
    \sqrt{\omega_k}
    (\overline{f_k(t)}
    \gamma_k
    +
    f_k(t)
    \gamma_k^\dagger
    ) =
    2
    \sum_{k=1}^{+\infty}
    \sqrt{\omega_k}
    \Re
    (    f_k(t)
    \gamma_k^\dagger
    )\\
\label{XQKL}
    & =
    2
    \sum_{k=1}^{+\infty}
    \sqrt{\omega_k}
    (\varphi_k(t)\xi_k + \psi_k(t)\eta_k)
    =
    2
    \sum_{k=1}^{+\infty}
    \sqrt{\omega_k}
    h_k(t)
    \zeta_k,
    \qquad
    0\< t\< T,
\end{align}
where the coefficients $\sqrt{\omega_k}\zeta_k$ involve the pairs of noncommuting positions and momenta $\xi_k$, $\eta_k$ from (\ref{zetak}), and (\ref{hk}) is used. The representation (\ref{XQKL}) is a quantum Karhunen-Loeve (QKL)
expansion of the process $X$ over the eigenbasis of its two-point CCR kernel (whose particular one-mode version with $n=2$ was discussed in Ref.~\refcite{VJP_2019}).

\section{Quadratic-Exponential Functionals  and the QKL Expansion}
\label{sec:quadro}

Consider an integral-of-quadratic function
of the quantum process $X$ over the time interval $[0,T]$:
\begin{equation}
\label{Q}
    Q
    :=
    \int_0^T
    X(t)^{\rT} X(t)
    \rd t
    =
    \int_0^T
    \sum_{k=1}^n
    X_k(t)^2
    \rd t,
\end{equation}
which is  a
positive semi-definite self-adjoint quantum variable.
A more general form $\int_0^T X(t)^{\rT} \Pi X(t) \rd t$
of such a function, specified by a real positive  semi-definite symmetric matrix $\Pi$ of order $n$, is reduced to (\ref{Q}) by replacing $X$ with its weighted version $\sqrt{\Pi} X$  which has the two-point CCR kernel $\sqrt{\Pi} \Lambda \sqrt{\Pi}$.

The minimization of the quadratic cost functional $\bE Q$ is used as
a performance criterion in linear quadratic Gaussian control and filtering problems\cite{MJ_2012,NJP_2009,ZJ_2012} for quantum systems, where $X$ consists of those system variables whose moderate mean square values are preferable. Here,
\begin{equation}
\label{bE}
    \bE \zeta
    :=
    \Tr(\rho \zeta)
\end{equation}
is the expectation of a quantum variable $\zeta$ over an underlying density operator $\rho$ on the space $\fH$. A more severe penalty on $Q$ in (\ref{Q}) is imposed by the quadratic-exponential functional\cite{VPJ_2018a} (QEF)
\begin{equation}
\label{Xi}
    \Xi
    :=
    \bE \re^{\frac{\theta}{2} Q}
\end{equation}
as a risk-sensitive cost, where $\theta>0$ is a risk sensitivity parameter. The QEF $\Xi$ allows the large deviations of $X$ to be quantified in terms of exponential upper bounds\cite{VPJ_2018a} on the tail probabilities for $Q$:
\begin{equation}
\label{tail}
    \bP([2\alpha, +\infty))
    \<
    \re^{-
    \sup_{\theta\>0}
    (
        \alpha\theta
        -
        \ln \Xi
    )},
    \qquad
    \alpha > 0,
\end{equation}
where $\bP(\cdot)$ is the probability distribution\cite{H_2001}  of the self-adjoint quantum variable $Q$, and the supremum is the Legendre transformation of $\ln \Xi$ as a function of $\theta\> 0$ (which vanishes at $\theta=0$ and is finite for sufficiently small $\theta>0$). Another useful property of $\Xi$ is its relation to upper bounds on $\bE Q$  in the presence of statistical uncertainty described in terms of  quantum relative entropy\cite{OW_2010,VPJ_2018b,YB_2009}, which will be discussed in Sec.~\ref{sec:worst} for quantum harmonic oscillators. These connections make the QEF $\Xi$ relevant to robustness properties of such systems  in the context of quantum control applications.

We will now apply the QKL expansion (\ref{XQKL})  of the quantum process $X$ from Sec.~\ref{sec:QKL} to computing the QEF. To this end, we will first represent the quadratic-exponential function $\re^{\frac{\theta}{2}Q}$ of $X$ in (\ref{Xi})   in a form which does not depend on a particular quantum state $\rho$.

\begin{lemma}
\label{lem:rand}
Under the conditions of Lemma~\ref{lem:posmom},   the quadratic-exponential function of the quantum process $X$ over the time interval $[0,T]$ in (\ref{Xi}) can be represented as
\begin{equation}
\label{Qexp}
  \re^{\frac{\theta}{2} Q}
  =
  \re^{-C}
  \bM \re^{\Sigma}.
\end{equation}
Here,
\begin{equation}
\label{C}
  C
  :=
  \frac{1}{2}
  \Tr \ln\cos (\theta \cL)
\end{equation}
is a nonnegative quantity, associated with the eigenfrequencies $\omega_k$ of the operator $\cL$ in (\ref{cL}). Also,
\begin{equation}
\label{Sig}
    \Sigma
    :=
  \sum_{k=1}^{+\infty}
  \sigma_k(\alpha_k \xi_k + \beta_k\eta_k)
\end{equation}
is a self-adjoint quantum variable which depends parametrically on
\begin{equation}
\label{uvvarsk}
    \sigma_k :=  \sqrt{2\tanh (\theta \omega_k)},
    \qquad
    k \in \mN,
\end{equation}
and on
mutually independent  standard normal (Gaussian with zero mean and unit variance) random variables $\alpha_k$, $\beta_k$, with $\bM(\cdot)$ in (\ref{Qexp}) denoting the classical expectation over
these random variables.
  \hfill$\square$
\end{lemma}
\begin{proof} By substituting the QKL expansion (\ref{XQKL}) into (\ref{Q}) and using (\ref{hh}), it follows that
\begin{equation}
\label{Q1}
      Q
     =
    4
    \sum_{j,k=1}^{+\infty}
    \sqrt{\omega_j\omega_k}\,
    \zeta_j^\rT
    \int_0^T
    h_j(t)^\rT
    h_k(t)
    \rd t
    \zeta_k
    = 2
    \sum_{k=1}^{+\infty}
    \omega_k
    \zeta_k^\rT
    \zeta_k
    =
    2
    \sum_{k=1}^{+\infty}
    \omega_k
    (\xi_k^2 + \eta_k^2).
\end{equation}
This reduction to a single series is due to the mutual orthogonality (\ref{reimort})  of the real and imaginary parts (\ref{reim}) of the eigenfunctions  (\ref{eigff}). A combination of (\ref{Q1}) with
Theorem~\ref{th:fact} leads to
$
    \re^{\frac{\theta}{2}Q}
    =
    \prod_{k=1}^{+\infty}
    \re^{\theta \omega_k (\xi_k^2 + \eta_k^2)}
    =
    \prod_{k=1}^{+\infty}
    \Big(
    \frac{1}{\cosh (\theta \omega_k)}
    \bM
    \re^{\sigma_k(\alpha_k \xi_k + \beta_k \eta_k)}
    \Big)
     =
    \re^{-\sum_{k=1}^{+\infty} \ln \cosh (\theta\omega_k)}
    \bM
    \re^{\sum_{k=1}^{+\infty} \sigma_k(\alpha_k \xi_k + \beta_k \eta_k) }
$,
which establishes (\ref{Qexp}) in view of (\ref{C})--(\ref{uvvarsk}).
Here, use is also made of the commutativity between the position-momentum pairs $(\xi_k,\eta_k)$ for different $k$.
The quantity $C$ in (\ref{C}) comes from the relations
\begin{equation}
\label{C1}
    \sum_{k=1}^{+\infty}\ln \cosh (\theta \omega_k)
  =
  \sum_{k=1}^{+\infty}\ln \cos(i \theta \omega_k)
  =
  \frac{1}{2}
  \Tr \ln\cos (\theta \cL)
\end{equation}
which take into account  the symmetry of the  functions $\cos$, $\cosh$ and the double multiplicity of the eigenfrequencies $\omega_k$. The operator trace in (\ref{C1}) inherits finiteness from the Hilbert-Schmidt norm of $\cL$ in (\ref{HSnorm}) since it is a particular case of (\ref{phiLtrace}) with a symmetric function $\phi(z):= \ln\cos (\theta z)$ satisfying $\phi(i\omega)= \ln \cosh (\theta \omega)\sim \frac{1}{2}\theta^2\omega^2 $, as $\omega\to 0$.
\end{proof}

\section{Auxiliary Classical Gaussian Random Process}
\label{sec:classproc}

The parameter randomisation in the form of averaging over the classical random sequence of $\alpha_k$, $\beta_k$ in Lemma~\ref{lem:rand} can be carried over  into continuous time. To this end, substitution of (\ref{xik_etak}) into (\ref{Sig}) leads to
\begin{equation}
\label{SZ}
    \Sigma
    =
  \sum_{k=1}^{+\infty}
    \frac{\sigma_k}{\sqrt{\omega_k}}
    \int_0^T
  (\alpha_k
  \varphi_k(t)
  +
  \beta_k
  \psi_k(t)
  )^\rT
  X(t)
  \rd t
  =
    \sqrt{\theta}
    \int_0^T
    X(t)^\rT
    \rd Z(t),
\end{equation}
where, in view of (\ref{uvvarsk}),
\begin{equation}
\label{Z}
    Z(t)
    :=
    \sum_{k=1}^{+\infty}
    \sqrt{2\tanhc (\theta\omega_k)}
    \int_0^t
    (\alpha_k \varphi_k(\tau)
    +
    \beta_k \psi_k(\tau)    )
    \rd \tau
    =
    \sum_{k=1}^{+\infty}
    \sqrt{\tanhc (\theta\omega_k)}\,
    H_k(t)
    {\begin{bmatrix}
      \alpha_k\\
      \beta_k
    \end{bmatrix}}
\end{equation}
on the time interval     $0\< t\< T$ is a classical $\mR^n$-valued zero-mean Gaussian random process  satisfying
\begin{equation}
\label{Z0}
    Z(0)=0.
\end{equation}
Here,  $H_k: [0,T] \to \mR^{n\x 2}$ are absolutely continuous functions,  associated with (\ref{hk}) by
\begin{equation}
\label{Hk}
  H_k(t)
  :=
  \sqrt{2}\int_0^t h_k(\tau) \rd \tau
  =
  \sqrt{2}\int_0^T \chi_{[0,t]}(\tau) h_k(\tau) \rd \tau,
  \qquad
  0\< t \< T,
\end{equation}
where $\chi_S(\cdot)$ is the indicator function of a set $S$,
and $\tanhc z := \tanc (-iz)$ is a hyperbolic version of $\tanc z := \frac{\tan z}{z}$ (extended to $1$ at $z=0$ by continuity). The covariance structure of the process (\ref{Z}) is as follows.

\begin{lemma}
\label{lem:Zcov}
Under the conditions of Lemmas~\ref{lem:posmom} and \ref{lem:rand},
the covariance operator of the incremented process $Z$ in (\ref{Z}) on the time interval $[0,T]$ is related to  $\cL$ in (\ref{cL}) by
\begin{equation}
\label{cK}
    \cK
    =
    \tanc (\theta\cL)
    =
    \tanhc(i\theta \cL)
\end{equation}
and satisfies
\begin{equation}
\label{cK01}
  0 \prec \cK \prec \cI,
\end{equation}
where $\cI$ is the identity operator on $L^2([0,T], \mR^n)$.\hfill$\square$
\end{lemma}
\begin{proof}
Since $\alpha_k$, $\beta_k$ are mutually independent standard  normal random variables, and hence, $\bM\Big({\small\begin{bmatrix}
  \alpha_j\\
  \beta_j
\end{bmatrix}}\begin{bmatrix}
  \alpha_k &   \beta_k
\end{bmatrix}\Big) = \delta_{jk}I_2$,  it follows from (\ref{Z}), (\ref{Hk}) that
\begin{align}
\nonumber
    \bM
    \Big(
    \Big(\int_0^T f(t)^\rT \rd Z(t)\Big)^2
    \Big)
    & =
    2
    \bM
    \Big(
    \Big(
    \sum_{k=1}^{+\infty}
    \sqrt{\tanhc (\theta\omega_k)}
    \int_0^T
    f(t)^\rT h_k(t)
    \rd t
    \begin{bmatrix}
      \alpha_k\\
      \beta_k
    \end{bmatrix}
    \Big)^2
    \Big)\\
\nonumber
    & =
    2
    \sum_{k=1}^{+\infty}
    \tanhc (\theta\omega_k)
    \int_{[0,T]^2}
    f(s)^\rT
    h_k(s)
    h_k(t)^\rT
    f(t)
    \rd s\rd t\\
\label{cKdef}
    & =
    2
    \sum_{k=1}^{+\infty}
    \tanhc (\theta\omega_k)
    \Big|
    \int_0^T
    h_k(t)^\rT
    f(t)
    \rd t
    \Big|^2=
    \bra f, \cK(f)\ket
\end{align}
for any $f \in L^2([0,T], \mR^n)$. Here, $\cK$ is the
covariance operator of the incremented process $Z$ with the kernel function
\begin{equation}
\label{Kst}
    K(s,t)
    :=
    2
    \sum_{k=1}^{+\infty}
    \tanhc (\theta\omega_k)
    h_k(s)
    h_k(t)^\rT,
    \qquad
    0 \< s,t \< T.
\end{equation}
This function can be represented in the form  (\ref{phiLambdaseries}) with a symmetric function $\phi(z):= \tanc(\theta z)$, and hence, (\ref{Kst}) is the kernel function of the operator $\tanc(\theta \cL)$, which implies (\ref{cK}). The inequalities (\ref{cK01}) follow from the fact that the eigenvalues (of double multiplicity) of the operator (\ref{cK}) satisfy $\tanc(\pm i\theta \omega_k) = \tanhc(\theta \omega_k) \in (0,1)$. These inequalities are closely related to the identity $    2
    \sum_{k=1}^{+\infty}
    \big|
    \int_0^T
    h_k(t)^\rT
    f(t)
    \rd t
    \big|^2 = \|f\|^2
$ for any $f \in L^2([0,T], \mR^n)$, whose terms are weighted by the eigenvalues of $\cK$ in (\ref{cKdef}).
\end{proof}

The operator inequalities (\ref{cK01}) are equivalent to
\begin{equation}
\label{iso}
    0 <
    \bra f, \cK(f)\ket
    <
    \|f\|^2
    =
    \bM\Big(\Big(\int_0^T f(t)^\rT \rd V(t)\Big)^2\Big)
\end{equation}
for any $f \in L^2([0,T], \mR^n)\setminus \{0\}$, where $V$ is a standard Wiener process\cite{KS_1991} in $\mR^n$ (whose increments have the identity covariance operator $\cI$), so that the last equality in (\ref{iso}) is the Ito isometry.
Since the functions (\ref{reim}) satisfy  (\ref{hh}), a combination of (\ref{Hk}) with the Plancherel identity leads to
\begin{equation}
\label{HH}
  \sum_{k=1}^{+\infty}
  H_k(s)H_k(t)^\rT
  =
  \bra
    \chi_{[0,s]},
    \chi_{[0,t]}
  \ket
  I_n
  =
  \min (s,t)I_n,
  \qquad
  0\< s,t\< T,
\end{equation}
which is the covariance function of the Wiener process $V$.  The process $Z$ in (\ref{Z}) has a different covariance function
\begin{equation}
\label{K}
  \bM(Z(s)Z(t)^\rT)
  =
  \sum_{k=1}^{+\infty}
  \tanhc (\theta\omega_k)\,
  H_k(s)H_k(t)^\rT,
  \qquad
  0\< s,t\< T,
\end{equation}
which is majorised by (\ref{HH}) in the sense that
$
    \bM (\bra f, Z\ket^2)
     =
    \int_{[0,T]^2}
    f(s)^\rT
    \bM(Z(s)Z(t)^\rT)
    f(t)
    \rd s\rd t
     =
      \sum_{k=1}^{+\infty}
  \tanhc (\theta\omega_k)
  \big|
  \int_0^T
  H_k(t)^\rT f(t)
  \rd t
  \big|^2
  \<
      \sum_{k=1}^{+\infty}
  \big|
  \int_0^T
  H_k(t)^\rT f(t)
  \rd t
  \big|^2
  =
    \int_{[0,T]^2}
    \min(s,t)f(s)^\rT f(t)\rd s\rd t
    =
    \bM (\bra f, V\ket^2)
$
for any $f \in L^2([0,T],\mR^n)$ since, as mentioned above, the function $\tanhc$ does not exceed $1$ on the real axis.

By using (\ref{SZ}), the randomised representation (\ref{Qexp}) can now be recast in continuous time as
\begin{equation}
\label{eQZ}
  \re^{\frac{\theta}{2}Q}
  =
  \re^{- \frac{1}{2}\Tr \ln \cos (\theta\cL)}
    \bM \re^{\sqrt{\theta}\int_0^T X(t)^\rT \rd Z(t) },
\end{equation}
where the averaging is over the classical Gaussian random process $Z$ from (\ref{Z}).
The general structure of this representation is reminiscent of the Doleans-Dade exponential\cite{DD_1970}.  However, the covariance function (\ref{K}) of the process $Z$ in (\ref{eQZ})
is associated with the commutator kernel $\Lambda$ in  (\ref{XXcomm}), and so also is the correction factor $\re^{- \frac{1}{2}\Tr \ln \cos (\theta\cL)}$. The latter (also see the proof of Theorem~\ref{th:fact}) comes from the Weyl CCRs (rather than the martingale property of the  Radon-Nikodym density process in Girsanov's theorem\cite{G_1960} on absolutely continuous change of measure for classical diffusion processes). In the limiting
classical case, when the commutator kernel $\Lambda$  vanishes (and so do the linear operator $\cL$ in (\ref{cL}) and its eigenfrequencies $\omega_k$), the covariance function (\ref{K}) reduces to (\ref{HH}), and $Z$ becomes the standard Wiener process.

The validity of (\ref{eQZ}) does not rely on a particular quantum state and employs only the commutation structure of the process $X$. Nevertheless, the relevance of this representation to computing the QEF (\ref{Xi})
is clarified by the following theorem which summarises the previous results.

\begin{theorem}
\label{th:MGF}
Suppose the conditions of Lemmas~\ref{lem:posmom} and \ref{lem:rand} are satisfied. Then  the QEF (\ref{Xi}) is related by
\begin{equation}
\label{XiZ}
  \Xi
    =
  \re^{- C}
  \bM
    \Omega(\sqrt{\theta}Z)
\end{equation}
to the quantity $C$ in (\ref{C}) and the auxiliary Gaussian random process $Z$ in (\ref{Z}) through a modified moment-generating functional (MGF)
\begin{equation}
\label{Psi}
    \Omega(z)
    :=
      \bE
    \re^{z(T)^\rT X(T)-\int_0^T z(t)^\rT\rd X(t)}
\end{equation}
of the quantum process $X$ over the time interval $[0,T]$, considered for functions
$z \in L^2([0,T],\mR^n)$ satisfying $z(0)=0$.  \hfill$\square$
\end{theorem}
\begin{proof}
In view of (\ref{C}), application of the quantum expectation (\ref{bE}) to  (\ref{eQZ}) allows (\ref{Xi}) to be represented as
\begin{equation}
\label{XiZ'}
  \Xi
  =
  \re^{- C}
  \bE
    \bM
    \re^{\sqrt{\theta}\int_0^T X(t)^\rT\rd Z(t)}
    =
  \re^{- C}
      \bM
  \bE
    \re^{\sqrt{\theta}\int_0^T X(t)^\rT\rd Z(t)},
\end{equation}
where use is made of the commutativity between the quantum and classical expectations $\bE(\cdot)$, $\bM(\cdot)$,  which describe the averaging over statistically independent quantum variables and auxiliary classical random variables. A combination of (\ref{XiZ'}) with the integration by parts $\int_0^T X(t)^\rT\rd Z(t) = Z(T)^\rT X(T) - \int_0^T Z(t)^\rT \rd X(t)$ in view of the initial condition (\ref{Z0}) leads to (\ref{XiZ}), where the MGF $\Omega$ from (\ref{Psi}) is evaluated   in a pathwise fashion at the rescaled version  $\sqrt{\theta}Z$ of the random process (\ref{Z}).
\end{proof}

The quantum state $\rho$ enters the representation (\ref{XiZ}) of the QEF $\Xi$ through the MGF $\Omega$ of the quantum process $X$, whose commutation kernel specifies the covariance structure of the auxiliary zero-mean Gaussian random process $Z$ through (\ref{cK}).

\section{QEF Representation for Multipoint Gaussian Quantum States}
\label{sec:gauss}

Note that (\ref{XiZ}) does not employ specific assumptions on the underlying quantum state $\rho$ except for the finiteness of exponential moments of the quantum process $X$ (in order for its MGF in (\ref{Psi}) to be well-defined). We will now specify this representation to the case when $X$ is in a multipoint Gaussian quantum state. Similarly to the classical Gaussian random processes\cite{GS_2004}, the Gaussian state for the quantum process $X$ is determined by its mean $\mu$ and real covariance kernel $P$ given by
\begin{equation}
\label{muP}
  \mu(t):= \bE X(t),
  \qquad
  P(s,t):= \Re \cov(X(s), X(t)),
  \qquad
  s,t\> 0,
\end{equation}
which are assumed to be continuous functions of time with values in $\mR^n$ and $\mR^{n\x n}$, respectively, with
\begin{equation}
\label{symm}
    P(s,t) = P(t,s)^\rT,
    \qquad
    s,t\> 0.
\end{equation}
Here,
\begin{equation}
\label{covX}
  \cov(X(s), X(t))
   :=
  \bE ((X(s)-\mu(s))(X(t)-\mu(t))^\rT)=
  P(s,t) + i\Lambda(s,t)
\end{equation}
is the quantum covariance function of the process $X$, with the imaginary part $\Lambda$ not depending on the quantum state. The kernel $P$ gives rise to a positive semi-definite self-adjoint operator $\cP$ acting on $L^2([0,T], \mC^n)$ as
\begin{equation}
\label{cP}
    \cP(f)(s):= \int_0^T P(s,t)f(t)\rd t,
    \qquad
    0\< s \< T.
\end{equation}
The positive semi-definiteness of $\cP$ follows from that of
the self-adjoint operator  $\cP+i\cL$ since
$
    \bra f, (\cP+i\cL)(f)\ket
     =
    \int_{[0,T]^2}
    f(s)^*
    \cov(X(s), X(t))
    f(t)
    \rd s\rd t
     =
    \bE(Y^\dagger Y)
    \> 0
$
for all     $
    f \in L^2([0,T],\mC^n)$, where     $
    Y:=
    \int_0^T f(t)^\rT (X(t)-\mu(t))
    \rd t$.
In the multipoint Gaussian quantum state, the quasi-characteristic functional\cite{CH_1971}  (QCF) of the process $X$ over the time interval $[0,T]$ takes the form
\begin{equation}
\label{QCF}
  \bE \re^{i\int_0^T u(t)^\rT X(t)\rd t}
  =
  \re^{i\bra\mu,u\ket  - \frac{1}{2} \bra u, \cP(u) \ket },
  \qquad
   u \in L^2([0,T],\mR^n).
\end{equation}
This QCF is identical to its classical Gaussian counterpart except that it comes with the stronger constraint $\cP+i\cL \succcurlyeq 0$.

\begin{theorem}
\label{th:main}
Suppose the conditions of Lemmas~\ref{lem:posmom} and \ref{lem:rand} are satisfied, and the quantum process $X$ over the time interval $[0,T]$ is in a multipoint Gaussian quantum state in the sense of (\ref{QCF}) with continuous mean $\mu$ and real covariance kernel $P$ in (\ref{muP}). Also, suppose the risk sensitivity parameter $\theta$ satisfies
\begin{equation}
\label{spec}
    \theta \br(\cP\cK) < 1,
\end{equation}
where $\br(\cdot)$ is the spectral radius, and  $\cP$, $\cK$ are the operators from (\ref{cP}), (\ref{cK}). Then the QEF (\ref{Xi}) can be found from
\begin{equation}
\label{XiZ2}
  \ln \Xi
    =
    -  \frac{1}{2}
  \Tr (\ln\cos (\theta\cL) + \ln (\cI - \theta \cP\cK))
  +
  \frac{\theta}{2} \bra \mu, (\cK(\cI-\theta \cP \cK)^{-1})(\mu)\ket.
\end{equation}
\hfill$\square$
\end{theorem}
\begin{proof}
Due to (\ref{QCF}), the MGF (\ref{Psi}) admits a closed-form representation, and its substitution into (\ref{XiZ}) yields
\begin{equation}
\label{XiZ1}
  \Xi
    =
    \re^{-C}
    \bM
  \re^{\sqrt{\theta} \int_0^T \mu(t)^\rT \rd Z(t) + \frac{\theta}{2}\int_{[0,T]^2} \rd Z(s)^\rT P(s,t)\rd Z(t)}.
\end{equation}
The right-hand side of (\ref{XiZ1}) does not involve quantum variables and is organised as a quadratic-exponential moment for the classical Gaussian random process $Z$ in (\ref{Z}). In order to compute it,
let $N$ be another independent auxiliary classical Gaussian random process with values in $\mR^n$ and the same mean $\mu$ and covariance function $P$ as the quantum process $X$ in (\ref{muP}):
$$
  \bM N(t) = \mu(t),
  \qquad
  \bM ((N(s)-\mu(s))(N(t)-\mu(t))^\rT) = P(s,t),
  \qquad
  s,t\> 0.
$$
Then by using the MGFs for Gaussian random processes and the tower property of iterated conditional classical expectations\cite{S_1996}, it follows that
\begin{align}
\nonumber
    \bM
  \re^{\sqrt{\theta} \int_0^T \mu(t)^\rT \rd Z(t) + \frac{\theta}{2}\int_{[0,T]^2} \rd Z(s)^\rT P(s,t)\rd Z(t)}
   & =
  \bM
    \bM
    \big(
  \re^{\sqrt{\theta}\int_0^T N(t)^\rT\rd Z(t)}
  \mid
  Z\big)\\
\nonumber
  & =
  \bM
  \re^{\sqrt{\theta}\int_0^T N(t)^\rT\rd Z(t)}
  =
  \bM
    \bM
    \big(
  \re^{\sqrt{\theta}\int_0^T N(t)^\rT\rd Z(t)}
  \mid
  N\big)\\
\nonumber
    & =
    \bM
  \re^{\frac{\theta}{2}\bra N, \cK(N)\ket }
  =
  \re^{-\frac{1}{2} \Tr \ln (\cI - \theta \cP \cK) + \frac{\theta}{2} \bra \mu, (\cK(\cI-\theta \cP \cK)^{-1})(\mu)\ket}\\
\label{MMM}
  & =
  \re^{\frac{\theta}{2} \bra \mu, (\cK(\cI-\theta \cP \cK)^{-1})(\mu)\ket}
  \prod_{k=1}^{+\infty}
  \frac{1}{\sqrt{1-\theta \lambda_k}},
\end{align}
where use is also made of  the covariance operator $\cK$  of the incremented process $Z$ in (\ref{cK}), (\ref{cKdef}) and the Fredholm determinant formula
(Theorem~3.10 on p.~36 of Ref.~\refcite{S_2005}, see also Ref.~\refcite{G_1994} and references therein).
Here,  $\lambda_k\>0$ are the eigenvalues of $\cP\cK$ which is a compact operator on $L^2([0,T], \mR^n)$ (isospectral to the positive semi-definite self-adjoint operator $\sqrt{\cK} \cP \sqrt{\cK}$) as the composition of a bounded operator ($\cK$) and a compact operator ($\cP$ in (\ref{cP})). For the validity of (\ref{MMM}), the condition
$    \br(\cP\cK) = \max_{k\in \mN} \lambda_k < \frac{1}{\theta}
$
has to be satisfied, which is equivalent to (\ref{spec}). In this case, the quadratic form $\bra \mu, (\cK(\cI-\theta \cP \cK)^{-1})(\mu)\ket$ in the mean $\mu$ (which is square integrable over $[0,T]$ due to its continuity) on the right-hand side of (\ref{MMM})  is finite and nonnegative since it is specified by a  positive definite self-adjoint operator satisfying
$
    \cK(\cI-\theta \cP \cK)^{-1}
     = \sqrt{\cK}(\cI-\theta \sqrt{\cK}\cP \sqrt{\cK})^{-1}\sqrt{\cK}
     \preccurlyeq
    \frac{1}{1-\theta \br(\cP\cK)}
    \cK
    \prec
    \frac{1}{1-\theta \br(\cP\cK)}
    \cI
$,
where use is also made of (\ref{cK01}). Furthermore, $-\Tr \ln (\cI - \theta \cP \cK) = -\sum_{k=1}^{+\infty} \ln(1-\theta \lambda_k) $ is also a finite nonnegative quantity due to (\ref{spec}), the asymptotic relation $-\ln(1-\theta \lambda)\sim \theta \lambda$ as $\lambda\to 0$,  and since the operator $\cP\cK$
(in view of its isospectrality to $\sqrt{\cK} \cP \sqrt{\cK} \succcurlyeq 0$) is of trace class\cite{B_1988}:
$
    \Tr (\cP\cK)
    =
    \sum_{k=1}^{+\infty}
    \lambda_k
    =
    \Tr (\sqrt{\cP} \cK \sqrt{\cP})
    \<
    \Tr \cP
    =
    \int_0^T
    \Tr P(t,t)\rd t < +\infty
$,
where (\ref{cK01}) is used again along with continuity of the real covariance kernel function $P$. Therefore, substitution of (\ref{C}), (\ref{MMM}) into (\ref{XiZ1}) leads to the relation (\ref{XiZ2}) for the QEF (\ref{Xi}).
\end{proof}

In the limiting classical case mentioned in Sec.~\ref{sec:classproc}, when $\cL = 0$ and $\cK=\cI$ in accordance with (\ref{cK}), the condition (\ref{spec}) reduces to $\theta<\frac{1}{\br(\cP)}$, and (\ref{XiZ2}) takes the form
$
  \ln \Xi
    =
    -  \frac{1}{2}
  \Tr \ln (\cI - \theta \cP\cK)
  +
  \frac{\theta}{2} \bra \mu, (\cI-\theta \cP )^{-1}(\mu)\ket
$.

\section{Infinite-Horizon QEF Growth Rate in the Stationary Gaussian Case}
\label{sec:stat}

We will now discuss the infinite-horizon asymptotic behaviour (as $T\to +\infty$)  of the QEF  (\ref{Xi}), computed in Theorem~\ref{th:main}, in the stationary Gaussian case.   It is assumed in what follows that the Gaussian quantum process $X$ has zero mean (for simplicity), while the real covariance and commutator kernels $P$, $\Lambda$  in (\ref{muP}), (\ref{XXcomm}) depend on the time difference:
\begin{equation}
\label{muPstat}
  \bE X(t) = 0,
  \qquad
  \Re \bE (X(s)X(t)^\rT) = P(s-t),
  \qquad
    [X(s), X(t)^\rT]
    =
    2i\Lambda(s-t),
\end{equation}
and hence, so also does the quantum covariance kernel of the process $X$ in (\ref{covX}):
\begin{equation}
\label{EXX}
    \bE (X(s)X(t)^\rT) = P(s-t) + i\Lambda(s-t),
    \qquad
    s,t\> 0.
\end{equation}
Here, the functions $P, \Lambda: \mR\to \mR^{n\x n}$ are assumed to be continuous  and absolutely integrable:
\begin{equation}
\label{int}
  \int_\mR
  \|P(\tau) + i\Lambda(\tau)\|_\rF
  \rd \tau
  =
  \int_\mR
  \sqrt{\|P(\tau)\|_\rF^2  + \|\Lambda(\tau)\|_\rF^2}
  \rd \tau
  < +\infty
\end{equation}
(any other matrix norm can also be used instead of $\|\cdot\|_\rF$
without affecting the integrability).
With a slight abuse of notation, the symmetry and antisymmetry properties (\ref{symm}), (\ref{skew}) of these kernels are equivalent to
$$
    P(\tau) = P(-\tau)^\rT,
    \qquad
    \Lambda(\tau) = -\Lambda(-\tau)^\rT,
    \qquad
    \tau \in \mR.
$$
The dependence on the risk sensitivity parameter $\theta$ and the time horizon $T$ will be  reinstated in the form of  subscripts (which were omitted previously for brevity), so that the operators (\ref{cP}), (\ref{cL}) take the form
\begin{equation}
\label{cPT_cLT}
    \cP_T(f)(s)
     = \int_0^T P(s-t)f(t)\rd t,
     \qquad
    \cL_T(f)(s)
    =
    \int_0^T
    \Lambda(s-t) f(t)
    \rd t,
    \qquad
    0\< s \< T,
\end{equation}
for any $f\in L^2([0,T], \mC^n)$ and $T>0$. Accordingly,  the operator (\ref{cK}) is given by
\begin{equation}
\label{cKT}
    \cK_{\theta,T}
    =
    \tanc (\theta\cL_T).
\end{equation}
In the case of $\mu=0$ being considered, the $\mu$-dependent term in (\ref{XiZ2}) vanishes, and this representation for the QEF $\Xi_{\theta,T}$ in (\ref{Xi}) reduces to the sum of ``trace-analytic''\cite{VP_2010} functionals:
\begin{equation}
\label{lnXi1}
  \ln \Xi_{\theta,T}
    =
    -  \frac{1}{2}
  \Tr (\varphi(\theta \cP_T \cK_{\theta,T}   ) + \psi(\theta\cL_T )),
\end{equation}
where
\begin{equation}
\label{phipsi}
  \varphi(z):= \ln(1-z),
  \qquad
  \psi(z):= \ln \cos z,
  \qquad
  z \in \mC,
\end{equation}
are holomorphic functions whose domains contain the spectra of the operators $\theta \cP_T \cK_{\theta,T}   $ (under the condition (\ref{spec})) and $\theta\cL_T $, at which these functions are evaluated.

We will now take into account the dependence of the   operators $\cP_T $, $\cL_T $ in (\ref{cPT_cLT}) (and the related operator $\cK_{\theta,T}  $ in (\ref{cKT})) on the time horizon  $T>0$. Each of them is organised as an integral operator $\cF_T$ on $L^2([0,T],\mC^n)$ whose kernel $F_T: [0,T]^2\to \mC^{n\x n}$ is obtained from a bounded continuous and absolutely integrable function $f: \mR\to \mC^{n \x n}$ (a shift-invariant kernel) as $F_T(s,t):= f(s-t)$ for all  $0\< s,t\< T$. Accordingly, the composition $\cH_T:= \cF_T \cG_T$ of such integral operators, whose kernel functions $F_T, G_T: [0,T]^2\to \mC^{n\x n}$ are generated by $f, g: \mR \to \mC^{n\x n}$ as described above,  is an integral operator whose kernel is an appropriately restricted convolution  $H_T(s,t):= \int_0^T F_T(s,u)G_T(u,t)\rd u = \int_0^T f(s-u)g(u-t)\rd u$ of the functions $f$, $g$ for all
$0\< s,t\< T$. If the operator $\cH_T$ is of trace class, then
$
    \Tr \cH_T
    =
    \int_0^T \Tr H_T(t,t)\rd t
    =
    \int_{[0,T]^2}
    \Tr (f(s-t)g(t-s))
    \rd s \rd t
$.
This relation extends to the rightward-ordered product $\cF_T := \rprod_{k=1}^q \cF_T^{(k)}$ of any number $q$ of such operators  with the kernel functions $F_T^{(k)}: [0,T]^2 \to \mC^{n\x n}$, generated by $f_k: \mR \to \mC^{n\x n}$ as described above, with $k=1, \ldots, q$, so that if the operator $\cF_T$ is of trace class, then
\begin{equation}
\label{TrFT}
    \Tr \cF_T
    =
    \int_{[0,T]^q}
    \Tr
    \rprod_{k=1}^q
    f_k (t_k - t_{k+1})
    \rd t_1\x \ldots \x \rd t_q,
\end{equation}
where $t_{q+1}:= t_1$.
Application of Lemma~6 from Appendix~C of Ref.~\refcite{VPJ_2018a}
to (\ref{TrFT}) leads to
\begin{equation}
\label{lim}
    \lim_{T\to +\infty}
    \Big(
        \frac{1}{T}
        \Tr \cF_T
    \Big)
    =
    \frac{1}{2\pi}
    \int_{\mR}
    \rprod_{k=1}^q
    \Phi_k(\lambda)
    \rd \lambda,
\end{equation}
where $\Phi_k(\lambda):= \int_\mR \re^{-i\lambda t }f_k(t)\rd t$ is the Fourier transform of the kernel function $f_k$. In turn, (\ref{lim}) extends to functions $\mC^q\ni z:= (z_j)_{1\< j\< q}\mapsto  h(z) := \sum_{k \in \mZ_+^q} c_k z^k \in \mC$ of $q$ complex variables with coefficients $c_k \in \mC$ and radii of convergence $r:=(r_j)_{1\<j\< q} \in \mR_+^q$ in the sense that $\sum_{k \in \mZ_+^q} |c_k| r^k<+\infty$, where the multiindex notation $z^k:= z_1^{k_1}\x \ldots \x z_q^{k_q}$ (and similarly for $r^k$) is used for any $q$-index $k:= (k_j)_{1\< j\< q} \in \mZ_+^q$, with $\mZ_+:= \{0,1,2,\ldots\}$ the set of nonnegative integers. For such a function $h$,
\begin{equation}
\label{lim1}
    \lim_{T\to +\infty}
    \Big(
        \frac{1}{T}
        \Tr h(\cF_T^{(1)}, \ldots, \cF_T^{(q)})
    \Big)
    =
    \frac{1}{2\pi}
    \int_{\mR}
    \Tr
    h(
    \Phi_1(\lambda),
    \ldots,
    \Phi_q(\lambda)
    )
    \rd \lambda,
\end{equation}
provided the same expansion of $h$  is applied to the noncommutative variables on both sides of (\ref{lim1}), and $\sup_{\lambda \in \mR} \|\Phi_j(\lambda)\| < r_j$ for all $j = 1, \ldots, q$, where $\|\cdot \|$ is the operator norm (the largest singular value) of a matrix.

The following theorem is concerned with the asymptotic growth rate of the quantity (\ref{lnXi1}), as $T\to +\infty$, and employs the Fourier transforms of the real covariance and commutator kernels in (\ref{muPstat}):
\begin{equation}
\label{Phi0_Psi0}
    \Phi(\lambda)
    :=
    \int_\mR \re^{-i\lambda t }
    P(t)
    \rd t,
    \qquad
    \Psi(\lambda)
    :=
    \int_\mR \re^{-i\lambda t }
    \Lambda(t)
    \rd t,
    \qquad
    \lambda \in \mR,
\end{equation}
which are well-defined continuous functions due to the integrability condition (\ref{int}).   Note that $\Phi(\lambda)$ is a complex positive semi-definite Hermitian  matrix, while  $\Psi(\lambda)$  is skew Hermitian for any $\lambda \in \mR$, with $\Phi(\lambda)+i\Psi(\lambda)$ being a complex positive semi-definite Hermitian  matrix as the Fourier transform  of the quantum covariance kernel $P+i\Lambda$ in (\ref{EXX}).

\begin{theorem}
\label{th:limXi}
Suppose the quantum process $X$ is stationary Gaussian with zero mean and a continuous quantum covariance function (\ref{EXX}) satisfying (\ref{int}). Also, suppose the operator $\cL_T$ in (\ref{cPT_cLT}) has no zero eigenvalues
for all sufficiently large time horizons $T>0$.
Furthermore, suppose the risk sensitivity parameter $\theta>0$ in (\ref{Xi}) satisfies
\begin{equation}
\label{spec1}
    \theta
    \sup_{\lambda \in \mR}
    \lambda_{\max}
    (
        \Phi(\lambda)
        \tanc
        (\theta \Psi(\lambda))
    )
    < 1,
\end{equation}
where $\lambda_{\max}(\cdot)$ is the largest eigenvalue, and the functions $\Phi$, $\Psi$ are given by (\ref{Phi0_Psi0}). Then the QEF $\Xi_{\theta,T}$, defined by  (\ref{Xi}), (\ref{Q}), has the following infinite-horizon growth rate:
\begin{equation}
\label{Ups}
    \Ups(\theta)
    :=
\lim_{T\to +\infty}
    \Big(
        \frac{1}{T}
        \ln \Xi_{\theta,T}
    \Big)
     =
    -
    \frac{1}{4\pi}
    \int_{\mR}
    \ln\det
    D_\theta(\lambda)
    \rd \lambda,
\end{equation}
where
\begin{equation}
\label{D}
    D_\theta(\lambda)
    :=
    \cos(
        \theta \Psi(\lambda)
    ) -
        \theta
        \Phi(\lambda)
        \sinc
        (\theta \Psi(\lambda)),
\end{equation}
and $\sinc z := \frac{\sin z}{z}$ (which is extended  as $\sinc 0 := 1$ by continuity).\hfill$\square$
\end{theorem}
\begin{proof}
In the case of one integral operator, the noncommutativity  issue does not arise, and (\ref{lim1}) is directly applicable to the second part of (\ref{lnXi1}) as
\begin{equation}
\label{psilim}
    \lim_{T\to +\infty}
    \Big(
        \frac{1}{T}
        \Tr \psi(\theta\cL_T )
    \Big)
    =
    \frac{1}{2\pi}
    \int_{\mR}
    \Tr \ln\cos(
        \theta \Psi(\lambda)
    )
    \rd \lambda
    =
    \frac{1}{2\pi}
    \int_{\mR}
    \ln\det \cos(
        \theta \Psi(\lambda)
    )
    \rd \lambda,
\end{equation}
where the function $\psi$ is given by (\ref{phipsi}), and use is made of the identity $\Tr \ln N = \ln\det N$ for square matrices $N$, along with the  Fourier transform (\ref{Phi0_Psi0})
of the commutator kernel $\Lambda$ in (\ref{muPstat}).
Now, in application of (\ref{lim1}) to the first part of (\ref{lnXi1}),  the function $\varphi$ from (\ref{phipsi}) is evaluated at the operator $\theta \cP_T  \cK_{\theta,T}  $ which involves  two noncommuting integral operators $\cP_T $, $\cL_T $ in (\ref{cPT_cLT}) and the related operator $\cK_{\theta,T}  $ in (\ref{cKT}) as
\begin{equation}
\label{phiPK}
    \varphi(\theta \cP_T  \cK_{\theta,T}  )
    =
    -
    \sum_{r=1}^{+\infty}
    \frac{1}{r}
    \theta^r
    (\cP_T  \cK_{\theta,T}   )^r
    =
    -
    \sum_{r=1}^{+\infty}
    \frac{1}{r}
    \theta^r
    \sum_{k_1, \ldots, k_r = 0}^{+\infty}
    \rprod_{j=1}^r
    \big(
    c_{k_j}
    \theta^{2k_j}
    \cP_T
    \cL_T ^{2k_j}
    \big)
\end{equation}
under the condition (\ref{spec}).
Here, use is made of the Maclaurin series expansion $\tanc z = \sum_{k=0}^{+\infty} c_k z^{2k}$ (with coefficients $c_k \in \mR$) in view of the symmetry of the $\tanc$  function.
Application of (\ref{lim1}) to (\ref{phiPK}) along with a dominated convergence argument
yields
\begin{align}
\nonumber
     \lim_{T\to +\infty}
    \Big(
        \frac{1}{T}
        \Tr \varphi(\theta\cP_T \cK_{\theta,T}  )
    \Big)
    & =
    -
    \frac{1}{2\pi}
    \sum_{r=1}^{+\infty}
    \frac{1}{r}
    \theta^r\!\!\!
    \sum_{k_1, \ldots, k_r = 0}^{+\infty}
    \int_\mR\!\!
    \Tr
    \rprod_{j=1}^r\!
    \big(
    c_{k_j}
    \theta^{2k_j}
    \Phi(\lambda)
    \Psi(\lambda)^{2k_j}
    \big)
    \rd \lambda\\
\nonumber
    & =
    \frac{1}{2\pi}
    \int_{\mR}
    \Tr
    \ln(
    I_n -
        \theta
        \Phi(\lambda)
        \tanc
        (\theta \Psi(\lambda))
    )
    \rd \lambda\\
\label{philim}
    & =
    \frac{1}{2\pi}
    \int_{\mR}
    \ln\det(
    I_n -
        \theta
        \Phi(\lambda)
        \tanc
        (\theta \Psi(\lambda))
    )
    \rd \lambda,
\end{align}
where
the Fourier transforms (\ref{Phi0_Psi0}) are used.  The limit relation (\ref{philim}) holds under the condition (\ref{spec1}) which is an infinite-horizon counterpart of (\ref{spec}) in  the frequency domain.
By combining (\ref{psilim}), (\ref{philim}), it follows that the quantity (\ref{lnXi1}) has the asymptotic growth rate
\begin{align}
\nonumber
\lim_{T\to +\infty}
    \Big(
        \frac{1}{T}
        \ln \Xi_{\theta,T}
    \Big)
    & =
    -
    \frac{1}{4\pi}
    \int_{\mR}
    \ln\det(
    I_n -
        \theta
        \Phi(\lambda)
        \tanc
        (\theta \Psi(\lambda))
    )
    \rd \lambda
    -
    \frac{1}{4\pi}
    \int_{\mR}
    \ln\det \cos(
        \theta \Psi(\lambda)
    )
    \rd \lambda\\
\label{Xilim}
    & =
    -
    \frac{1}{4\pi}
    \int_{\mR}
    \ln\det(
    \cos(
        \theta \Psi(\lambda)
    ) -
        \theta
        \Phi(\lambda)
        \sinc
        (\theta \Psi(\lambda))
    )
    \rd \lambda,
\end{align}
where the identity
$\tanc z \cos z = \sinc z$ is applied to the matrix $\theta \Psi(\lambda)$. In view of (\ref{D}), the relation (\ref{Xilim}) establishes (\ref{Ups}).
\end{proof}

Under the condition (\ref{spec1}), the integrand $\ln\det D_\theta(\lambda)$ in (\ref{Ups}) is a nonpositive-valued  symmetric function of the frequency $\lambda$.
In the limiting classical case, when $\Lambda = 0$ (and hence, $\Psi=0$ in (\ref{Phi0_Psi0})),  so that $X$ is a stationary Gaussian random process\cite{GS_2004}  with zero mean and the spectral density $\Phi$ in (\ref{Phi0_Psi0}),  the condition (\ref{spec1}) takes the form
\begin{equation}
\label{class}
    \theta
    <
    \theta_*:=
    \frac{1}{\sup_{\lambda \in \mR}
    \lambda_{\max}
    (
        \Phi(\lambda)
    )},
\end{equation}
and the right-hand side of (\ref{Ups}) reduces to
\begin{equation}
\label{V}
    \Ups_*(\theta)
    :=
    -
    \frac{1}{4\pi}
    \int_{\mR}
    \ln\det(
    I_n
 -
        \theta
        \Phi(\lambda)
    )
    \rd \lambda
\end{equation}
in view of (\ref{D}).

Returning to the quantum case, we note that
the QEF growth rate (\ref{Ups}), in contrast to its classical counterpart, depends  on both functions $\Phi$, $\Psi$ which form the ``quantum spectral density'' $\Phi + i\Psi$ of the stationary Gaussian quantum process $X$. Furthermore, the condition (\ref{spec1}) is substantially nonlinear with respect to  $\theta$ and, unlike $\theta_*$ in (\ref{class}),  does not admit a closed-form representation. However, since $\tanc$ on the imaginary axis (that is, $\tanhc$ on the real axis) takes values in $(0,1]$, then
$
    \lambda_{\max}
    (
        \Phi(\lambda)
        \tanc
        (\theta \Psi(\lambda))
    )
    =
    \lambda_{\max}
    (
        \sqrt{\tanc
        (\theta \Psi(\lambda))}
        \Phi(\lambda)
        \sqrt{\tanc
        (\theta \Psi(\lambda))}
    )
    \< \lambda_{\max}(\Phi(\lambda))
$
for any $\lambda \in \mR$ by the isospectrality argument used in Sec.~\ref{sec:gauss}, and hence, the fulfillment of the classical constraint (\ref{class}) ensures (\ref{spec1}).

Also note that a combination of the QEF growth rate (\ref{Ups}) with the upper bounds (\ref{tail}) on the tail distribution of the quantum variable $Q_T$ in (\ref{Q}) leads to their infinite-horizon asymptotic version:
\begin{equation}
\label{supP}
    \limsup_{T\to +\infty}
    \Big(
        \frac{1}{T}
        \ln
        \bP_T([2\alpha T, +\infty))
    \Big)
    \<
    -
    \sup_{\theta\>0}
    (
        \alpha\theta
        -
        \Ups(\theta)
    ),
    \qquad
    \alpha >0,
\end{equation}
where $\bP_T(\cdot)$ is the probability distribution of $Q_T$. The Legendre transformation of $\Ups$ on the right-hand side of (\ref{supP}) requires techniques for computing the functional (\ref{Ups}) at different values of $\theta$.

\section{Evaluation of the QEF Growth Rate Using a Homotopy Technique}
\label{sec:comp}

The QEF growth rate $\Ups(\theta)$ in (\ref{Ups}) can be computed by the following  technique, similar to the homotopy methods for numerical solution of parameter dependent algebraic equations\cite{MB_1985}.
To this end, we will use a function $U_\theta: \mR \to \mC^{n\x n}$ defined by
\begin{equation}
\label{U}
  U_\theta(\lambda):= -D_\theta(\lambda)^{-1}\d_\theta D_\theta(\lambda)
\end{equation}
as a ``logarithmic derivative'' of the function $D_\theta$ from  (\ref{D}) with respect to $\theta >0$ satisfying (\ref{spec1}) (which ensures that $\det D_\theta (\lambda)\ne 0$ for all $\lambda \in \mR$).

\begin{theorem}
\label{th:diff}
Under the conditions of Theorem~\ref{th:limXi}, the QEF growth rate $\Ups(\theta)$ in  (\ref{Ups}) satisfies the ODE\footnote{more precisely, an integro-differential equation}
\begin{equation}
\label{Ups'}
  \Ups'(\theta)
  =
    \frac{1}{4\pi}
    \int_{\mR}
    \Tr U_\theta(\lambda)
    \rd \lambda,
\end{equation}
with the initial condition $\Ups(0)=0$. Here, the function (\ref{U}) is computed as
\begin{equation}
\label{U1}
U_\theta
=
\Psi
    (    \Psi\cos(
        \theta \Psi
        ) -
        \Phi
        \sin
        (\theta \Psi)
        )^{-1}
        (\Phi \cos(\theta \Psi)
    +\Psi\sin(\theta \Psi)
    )
\end{equation}
(the argument $\lambda$ is omitted for brevity),
takes values in the subspace of Hermitian matrices of order $n$ and satisfies a Riccati equation
\begin{equation}
\label{U'}
  \d_\theta U_\theta = \Psi^2 + U_\theta^2
\end{equation}
at any frequency $\lambda\in \mR$,
with the initial condition $U_0 = \Phi$ from (\ref{Phi0_Psi0}).\hfill$\square$
\end{theorem}
\begin{proof}
The relation (\ref{Ups'}) is obtained from (\ref{Ups}), (\ref{U}) by applying the identity $(\ln\det N)' = \Tr (N^{-1}N')$, where $(\cdot)':= \d_\theta(\cdot)$, so that $(\ln\det D_\theta)' = -\Tr U_\theta$.   With $\frac{\sin(\theta z)}{z}$ extended by continuity to $\theta$ at $z=0$,
the function $D_\theta$ in (\ref{D})  is represented as
\begin{equation}
\label{D1}
    D_\theta
    =
    \cos(
        \theta \Psi
    ) -
        \Phi
        \Psi^{-1}
        \sin
        (\theta \Psi)
\end{equation}
for any $\lambda\in \mR$, and hence, its  differentiation with respect to $\theta$ yields
\begin{equation}
\label{D'}
    D_\theta'
    =
    -\Psi
    \sin(
        \theta \Psi
    )
    -
    \Phi
        \cos
        (\theta \Psi).
\end{equation}
Substitution of (\ref{D1}), (\ref{D'}) into (\ref{U}) leads to (\ref{U1}). By differentiating  (\ref{D'}) in $\theta$, it follows that (\ref{D1})
satisfies the linear second-order  ODE
\begin{equation}
\label{D''}
  D_\theta''
  =
    -\Psi^2
    \cos(
        \theta \Psi
    )
    +
    \Phi
    \Psi
        \sin
        (\theta \Psi)
    =
  - D_\theta \Psi^2,
\end{equation}
with the initial conditions $D_0 = I_n$, $D_0' = -\Phi$. Therefore, the differentiation of (\ref{U}) leads to
$    U_\theta'
    =
    -D_\theta^{-1}D_\theta''
    +D_\theta^{-1}D_\theta' D_\theta^{-1}D_\theta'
    =
    \Psi^2 + U_\theta^2
$,
where the relation $(N^{-1})' = -N^{-1}N'N^{-1}$ is used along with (\ref{D''}), thus establishing (\ref{U'}).  The solution of (\ref{U'}) inherits the Hermitian property from its initial condition $U_0 = \Phi$, since  $\Psi(\lambda) = -\Psi(\lambda)^*$ in (\ref{Phi0_Psi0}) for any $\lambda\in \mR$, and $(N^2)^* = (\pm N)^2 = N^2$
for Hermitian or skew  Hermitian matrices $N$, respectively.  Regardless of the initial value problem for the ODE (\ref{U'}),   the Hermitian property $U_\theta(\lambda)=U_\theta(\lambda)^*$ can also be obtained directly from (\ref{U}),  (\ref{D1}), (\ref{D'}) as
$
    D_\theta(U_\theta - U_\theta^*)D_\theta^*
     = D_\theta D_\theta'^*-D_\theta' D_\theta^*
    =
    (    c -
        \Phi
        \Psi^{-1}
        s
    )
    (    -\Psi
    s
    -
    \Phi
        c
    )^*
            +
    (\Psi
    s
    +
        \Phi
        c)
        (    c -
        \Phi
        \Psi^{-1}
        s
)^*
    =
    (
        \Phi
        \Psi^{-1}
        s
        -c
    )
    (\Psi
    s
    +
    c
    \Phi
     )
            +
    (\Psi
    s
    +
        \Phi
        c)
        (    c -
        \Psi^{-1}
        s
        \Phi
)
    =
    \Phi s^2 + \Phi\Psi^{-1}s c\Phi-c\Psi s-c^2\Phi
     + \Psi sc-s^2 \Phi+\Phi c^2-\Phi c \Psi^{-1}s \Phi
    =
    \Phi(c^2 + s^2) -  (c^2 + s^2)\Phi  = 0
$,
where the matrices $c:= \cos(\theta \Psi) = c^*$ and $s:= \sin(\theta \Psi) = -s^*$ commute between themselves and with $\Psi = -\Psi^*$ (but not necessarily with $\Phi = \Phi^*$) and satisfy $c^2 + s^2 = I_n$.
\end{proof}

The relation (\ref{U}), which links the quadratically nonlinear ODE (\ref{U'}) with the linear ODE (\ref{D''}), can be regarded as a matrix-valued analogue of the Hopf-Cole transformation\cite{C_1951,H_1950} converting the viscous Burgers equation to the heat (or diffusion) equation.  We also mention an analogy between (\ref{U}) and the logarithmic transformation in the context of dynamic programming
equations for stochastic control\cite{F_1982} (see also Ref.~\refcite{VP_2010}).

The right-hand side of (\ref{Ups'}) can be evaluated by numerical integration over the frequency axis and used for computing (\ref{Ups}) as
$
    \Ups(\theta)
    =
    \int_0^\theta \Ups'(v)\rd v
    =
    \frac{1}{4\pi}
    \int_{\mR \x [0,\theta]}
    \Tr U_v(\lambda)
    \rd \lambda \rd v
$.
In particular, (\ref{Ups'}) yields
$
    \Ups'(0)
    =
    \frac{1}{4\pi}
    \int_\mR
    \Tr \Phi(\lambda)
    \rd \lambda
    =
    \frac{1}{2}
    \bE(X(0)^\rT X(0))
    =
    \frac{1}{2}
    \lim_{T\to +\infty}
    \big(
        \frac{1}{T}
        \bE Q_T
    \big)
$,
which, in accordance with (\ref{Q}),  (\ref{Xi}),  is related to the mean square cost functional growth rate  for the process $X$.
In addition to its role for the computation of $\Ups$, the function $\Ups'$ admits the following representation (see also Theorem~1 of Ref.~\refcite{VPJ_2018a}):
\begin{equation}
\label{Ups'1}
  \Ups'(\theta)
  =
  \frac{1}{2}
  \lim_{T\to +\infty}
  \Big(
    \frac{1}{T}
    \bE_{\theta, T} Q_T
  \Big),
\end{equation}
where $\bE_{\theta, T} \zeta := \Tr(\rho_{\theta, T}\zeta)$ is the quantum expectation over a modified density operator $\rho_{\theta, T}:= \frac{1}{\Xi_{\theta, T}}\re^{\frac{\theta}{4}Q_T} \rho \re^{\frac{\theta}{4}Q_T}$.  Therefore, (\ref{Ups'1}) relates $\Ups'$ to the asymptotic growth rate of the weighted average of the quantum variable $Q_T$ in (\ref{Q}) rather than its exponential moment.

If $\Phi(\lambda)$ is a rational function  of the frequency $\lambda$, then so also is $\det (I_n - \theta \Phi(\lambda))$, and the integral in the classical counterpart (\ref{V}) of the QEF growth rate can be evaluated using the residue theorem\cite{S_1992}. This observation can be combined with the Maclaurin series expansions of the trigonometric functions, which allows (\ref{D}) to be approximated as
\begin{equation}
\label{Dprox}
    D_\theta
    =
    I_n - \frac{1}{2}\theta^2 \Psi^2 - \theta \Phi \Big(I_n - \frac{1}{6}\theta^2 \Psi^2\Big)
    +
    o(\theta^3)
    =
    I_n - \theta \Phi
    -
    \frac{1}{2}\theta^2
    \Big(I_n - \frac{\theta}{3}\Phi\Big)\Psi^2
    +
    o(\theta^3),
    \qquad
    {\rm as}\ \theta \to 0.
\end{equation}
Substitution of (\ref{Dprox}) into (\ref{Ups}) leads to the approximate computation of the QEF growth rate as a perturbation of its classical counterpart (\ref{V}):
\begin{equation}
\label{VUps}
  \Ups(\theta)
  =
  \Ups_*(\theta)
  +
  \frac{\theta^2}{8\pi}
  \int_\mR
  \Tr
  \Big((I_n - \theta\Phi(\lambda))^{-1}\Big(I_n - \frac{\theta}{3}\Phi(\lambda)\Big)\Psi(\lambda)^2\Big)
  \rd \lambda
  + o(\theta^3),
    \quad
    {\rm as}\
    \theta \to 0.
\end{equation}
If $\Psi$ is also a rational function, then
the integrand in (\ref{VUps}) is a rational function whose continuation to the closed right half-plane $\{s\in \mC: \Re s \> 0\}$ has no poles on the imaginary axis under the condition (\ref{class}). This makes the correction term also amenable to calculation via the residues in the open right-half plane. Since $\Psi(\lambda)^2\preccurlyeq 0$ for all $\lambda \in \mR$ (due to $\Psi(\lambda)^*= -\Psi(\lambda)$ as mentioned before),  and the function $\Psi$ is continuous and is not identically zero, the relation (\ref{VUps}) also implies that $\Ups(\theta)< \Ups_*(\theta)$ for all sufficiently small $\theta >0$.

\section{Application to Open Quantum Harmonic Oscillators}
\label{sec:sys}

The case of rational functions $\Phi$, $\Psi$ in (\ref{Phi0_Psi0}), mentioned in the previous section,  arises, for example,  in the context of an open quantum harmonic oscillator (OQHO) driven by external bosonic fields. The latter are modelled by a multichannel quantum Wiener process $W:=(W_k)_{1\< k \< m}$ consisting of an even number of time-varying self-adjoint operators $W_1(t), \ldots, W_m(t)$ on a symmetric Fock space\cite{P_1992} $\fF$. For any $t\>0$ and $k=1, \ldots, m$, the operator $W_k(t)$ acts on a subspace $\fF_t$ of $\fF$, with  the
increasing family $(\fF_t)_{t\> 0}$  forming a filtration in accordance with the continuous tensor-product structure\cite{PS_1972}  of $\fF$. The component quantum Wiener processes $W_1, \ldots, W_m$ satisfy the two-point CCRs
\begin{equation}
\label{WWcomm}
    [W(s), W(t)^\rT]
     =
    2i\min(s,t)J,
    \qquad
    s,t\>0,
\end{equation}
where
    $J:=  \bJ \ox I_{m/2}$
is an orthogonal real  antisymmetric matrix (so that $J^2 = -I_m$), with the matrix $\bJ$ given by (\ref{bJ}).
The OQHO under consideration
is endowed with an even number of time-varying  self-adjoint quantum variables $\cX _1(t), \ldots, \cX _\nu(t)$ acting on the subspace
\begin{equation}
\label{fHt}
    \fH_t:= \fH_0 \ox \fF_t
\end{equation}
of the system-field tensor-product space $\fH:= \fH_0 \ox \fF$. Here, $\fH_0$ is a complex separable Hilbert space for the action of the initial system variables $\cX _1(0), \ldots, \cX _\nu(0)$. The vector     $\cX :=(\cX _k)_{1\< k\< \nu}$ of the system variables satisfies the Weyl CCRs\cite{F_1989}
\begin{equation}
\label{Weyl}
  \re^{i(u+v)^\rT \cX (t)}
  =
  \re^{iu^\rT\Theta v}
    \re^{iu^\rT \cX (t)}
      \re^{iv^\rT \cX (t)},
      \qquad
      u,v\in \mR^\nu,\
      \
      t\> 0,
\end{equation}
specified by a constant matrix $\Theta = -\Theta^{\rT} \in \mR^{\nu\x \nu}$.
The infinitesimal Heisenberg form of (\ref{Weyl}), similar to (\ref{Xcomm}), is  given by
\begin{equation}
\label{XCCR}
    [\cX (t),\cX (t)^\rT] = 2i \Theta.
\end{equation}
The evolution of the OQHO is governed by a linear Hudson-Parthasarathy QSDE\cite{HP_1984,P_1992}
\begin{equation}
\label{dX}
    \rd \cX  = A\cX  \rd t + B\rd W,
\end{equation}
driven by the quantum Wiener process $W$ described above. Here, the matrices $A \in \mR^{\nu\x \nu}$, $B\in \mR^{\nu\x m}$ are parameterised as
\begin{equation}
\label{AB}
    A = 2\Theta (R + M^\rT JM),
     \qquad
     B = 2\Theta M^\rT
\end{equation}
by the energy and coupling matrices $R = R^\rT \in \mR^{\nu\x \nu}$, $M \in \mR^{m\x \nu}$ which specify the system Hamiltonian $\frac{1}{2} \cX ^\rT R \cX $ and the vector $M\cX $ of $m$ system-field coupling operators. Due to their specific structure (\ref{AB}), the matrices $A$, $B$ satisfy the physical realizability (PR) condition\cite{JNP_2008}
\begin{equation}
\label{PR}
    A \Theta + \Theta A^\rT +  \mho = 0,
    \qquad
    \mho := BJB^\rT,
\end{equation}
which is an algebraic manifestation of the CCR matrix preservation in (\ref{XCCR}).
We will be concerned with the QEF (\ref{Xi}) for a quantum process $X:= (X_k)_{1\< k \< n}$ of time-varying self-adjoint operators $X_1, \ldots, X_n$ on $\fH$ related to the system variables of the OQHO by
\begin{equation}
\label{X}
    X:= S \cX,
\end{equation}
where $S \in \mR^{n\x \nu}$ is a given weighting matrix, which specifies the relative importance of the system variables. Since $S$ enters (\ref{Q}) only through the matrix $S^\rT S$ as $    Q_T
    =
    \int_0^T
    \cX(t)^{\rT} S^\rT S \cX(t)
    \rd t
$, then it can be assumed, without loss of generality, that $S$ is of full row rank:
\begin{equation}
\label{rank}
  \rank S = n \< \nu.
\end{equation}
The process $X$ in (\ref{X}) satisfies the two-point CCRs of the last equality  in (\ref{muPstat}), with the commutator kernel computed as\cite{VPJ_2018a}
\begin{equation}
\label{Lambda}
    \Lambda(\tau)
     :=
    \left\{
    {\begin{matrix}
    S\re^{\tau A}\Theta S^\rT& {\rm if}\  \tau\> 0\\
    S\Theta\re^{-\tau A^{\rT}}S^\rT & {\rm if}\  \tau< 0\\
    \end{matrix}}
    \right.,
    \qquad
    \tau \in \mR,
\end{equation}
so that its one-point CCR matrix is $\Lambda(0) = S\Theta S^\rT$, in accordance with (\ref{XCCR}).
The following theorem relates the corresponding integral operator $\cL_T$ in (\ref{cPT_cLT}) to a boundary value problem.

\begin{lemma}
\label{lem:BVP}
Suppose the quantum process $X$ in  (\ref{X}),  with the CCR kernel $\Lambda$ in  (\ref{Lambda}),  is  associated with the system variables of the OQHO described by   (\ref{WWcomm})--(\ref{AB}). Then for any time horizon $T>0$, the operator $\cL_T$ in (\ref{cPT_cLT})  maps a function $f\in L^2([0,T], \mC^n)$  to
\begin{equation}
\label{cLTS}
  \cL_T(f)(t) = S \cC h(t),
  \qquad
  0\< t \< T,
\end{equation}
where $h: [0,T]\to \mC^{2\nu}$ is an absolutely continuous function,  which is split into $h_{\pm}: [0,T]\to \mC^\nu$ as
\begin{equation}
\label{hhh}
    h(t)
    :=
    \begin{bmatrix}
      h_+(t)\\
      h_-(t)
    \end{bmatrix}
\end{equation}
and satisfies the ODE
\begin{equation}
\label{BVP}
  \dot{h}(t)
  =
  \cA h(t) + \cB S^\rT f(t),
  \qquad
  0 \< t\< T
\end{equation}
(with $\dot{(\, )}$ the time derivative)
subject to the boundary conditions
\begin{equation}
\label{h+-0}
    h_+(0) = h_-(T) = 0.
\end{equation}
Here, the matrices $\cA^{2\nu\x 2\nu}$, $\cB\in \mR^{2\nu \x \nu}$, $\cC \in \mR^{\nu\x 2\nu}$ are given by
\begin{equation}
\label{cABC}
    \cA
    :=
  \begin{bmatrix}
    A & 0\\
    0 & -A^\rT
  \end{bmatrix},
  \qquad
  \cB :=
  \begin{bmatrix}
    \Theta \\
    -I_\nu
  \end{bmatrix},
  \qquad
      \cC
    :=
  \begin{bmatrix}
    I_\nu & \Theta
  \end{bmatrix}.
\end{equation}
Every eigenfunction of $\cL_T$ with a nonzero eigenvalue is infinitely differentiable.
\hfill$\square$
\end{lemma}
\begin{proof}
Due to the structure of the commutator kernel $\Lambda$ in  (\ref{Lambda}), the operator $\cL_T$ in (\ref{cPT_cLT}) acts from $L^2([0,T],\mC^n)$ to the space of twice differentiable functions with square integrable second-order derivatives. Indeed, similarly to the Wiener-Hopf method\cite{GK_1958}, the image $g=\cL_T(f)$ of $f\in L^2([0,T],\mC^n)$ can be represented as
\begin{equation}
\label{gh}
  g(t) =
  S \cC
  h(t),
  \qquad
  0\< t\< T,
\end{equation}
in terms of the matrices $S$,  $\cC$ in (\ref{X}), (\ref{cABC}) and the function $h$ in (\ref{hhh}),  given by
\begin{equation}
\label{h+-}
    \qquad
    h_+(t)
    :=
    \int_0^t
    \re^{(t-\tau)A}\Theta S^\rT
    f(\tau)\rd \tau,
    \qquad
    h_-(t)
    :=
    \int_t^T
    \re^{(\tau-t)A^\rT}
    S^\rT
    f(\tau)\rd \tau
\end{equation}
and hence, satisfying the boundary conditions (\ref{h+-0}).
By using the matrix exponential structure of the kernel functions again, differentiation of (\ref{h+-})   yields
\begin{equation}
\label{h+-'}
    \dot{h}_+ = Ah_+ + \Theta S^\rT f,
    \qquad
    \dot{h}_- = - A^\rT h_--S^\rT f.
\end{equation}
In view of (\ref{hhh}), the ODEs (\ref{h+-'}) can be assembled into (\ref{BVP}),   with the matrices $\cA$, $\cB$ given by (\ref{cABC}), thus establishing (\ref{cLTS}).
Since $\cC \cB = 0$, the derivative of $g$ in (\ref{gh}), computed as
\begin{equation}
\label{gdot}
    \dot{g}
    =
    S\cC (\cA h + \cB S^\rT f) = S\cC\cA h
\end{equation}
by using (\ref{BVP}), is absolutely continuous, and  hence, the second derivative
\begin{equation}
\label{gddot}
    \ddot{g}
    =
    S\cC\cA (\cA h + \cB S^\rT f)
    =
    S\cC\cA^2 h + S\cC \cA \cB S^\rT f
\end{equation}
is square integrable over $[0,T]$ for any $f \in L^2([0,T], \mC^n)$. Therefore, if $f$ is an eigenfunction of $\cL_T$ with an  eigenfrequency $\omega>0$, then $f = \frac{1}{i\omega} g$ inherits this property (of being twice differentiable and having a square integrable second derivative), and hence, so also does the right-hand side of (\ref{gddot}). Successive differentiation of both sides of (\ref{gddot}) and application of induction leads to infinite differentiability of the eigenfunction $f$.
\end{proof}

Lemma~\ref{lem:BVP} represents the operator $\cL_T$ in (\ref{cPT_cLT})  as an auxiliary dynamical system on the time interval $[0,T]$ with the input $f$, output $g = \cL_T(f)$ in (\ref{cLTS}) and internal state $h$ governed by (\ref{BVP}). Accordingly, an appropriate restriction of the CCR kernel $\Lambda$ in (\ref{Lambda}) is the Green function for the BVP (\ref{BVP}), (\ref{h+-0}).
The presence of the boundary constraints (\ref{h+-0}) on the internal state of this system complicates analysis of its zero dynamics and development of conditions for $\cL_T$ not having  zero eigenvalues. The full row rank condition (\ref{rank}) on $S$ is necessary for $\ker \cL_T = \{0\}$. Indeed, otherwise, $\ker (S^\rT)\ne \{0\}$, and any nonzero function $f \in L^2([0,T], \ker (S^\rT))$ satisfies $\cL_T(f) =0$, whereby $\ker \cL_T \ne \{0\}$. Sufficient conditions  are provided below.

\begin{theorem}
\label{th:non0}
Suppose the assumptions of Lemma~\ref{lem:BVP} are satisfied, and the matrix $\mho = -\mho^\rT  \in \mR^{\nu\x \nu}$ in (\ref{PR}) is nonsingular:
\begin{equation}
\label{BJB}
  \det \mho \ne 0.
\end{equation}
Then for any time horizon $T>0$, the image of a function $f\in L^2([0,T], \mC^n)$ under the operator $\cL_T$ in (\ref{cPT_cLT}) can be represented as
\begin{equation}
\label{out}
  \cL_T(f)(t) = Sz(t),
  \qquad
  0\< t \< T,
\end{equation}
where $z:[0,T]\to \mC^\nu$ is a twice differentiable function satisfying the second-order ODE
\begin{equation}
\label{z''}
  \ddot{z}+(\mho A^\rT \mho^{-1}-A) \dot{z}-\mho A^\rT \mho^{-1} A z =-\mho S^\rT f,
\end{equation}
with the boundary conditions
\begin{equation}
\label{z+0_z-T}
    (I_\nu+\Theta \mho^{-1} A)z(0) = \Theta \mho^{-1} \dot{z}(0),
    \qquad
    \dot{z}(T)  = A z(T).
\end{equation}
If, in addition to the above assumptions, the weighting matrix $S$ in (\ref{X}) is square and nonsingular,
\begin{equation}
\label{detS}
  n=\nu, \qquad \det S\ne 0,
\end{equation}
then the operator $\cL_T$ has no zero eigenvalues for any $T>0$.
\hfill$\square$
\end{theorem}
\begin{proof}
In view of (\ref{cLTS}),  the relation (\ref{out}) holds with the function $z$ defined in terms of (\ref{hhh}) as
\begin{equation}
\label{zh}
  z:= \cC h
\end{equation}
and inheriting the absolute continuity from $h$. Similarly to (\ref{gdot}), differentiation of (\ref{zh}) leads to
\begin{equation}
\label{zdot}
    \dot{z}
    =
    \cC (\cA h + \cB S^\rT f) = \cC\cA h,
\end{equation}
where use is also made of (\ref{BVP}).
A combination of (\ref{zh}), (\ref{zdot}) with (\ref{cABC}) yields
\begin{equation}
\label{zzh}
    \begin{bmatrix}
      z\\
      \dot{z}
    \end{bmatrix}
    =
    \begin{bmatrix}
      \cC\\
      \cC\cA
    \end{bmatrix}
    h
    =
    \begin{bmatrix}
      I_\nu & \Theta\\
      A & -\Theta A^\rT
    \end{bmatrix}
    h.
\end{equation}
By using the Schur complement\cite{HJ_2007} of the block $I_\nu$ together with the PR condition (\ref{PR}) and the assumption (\ref{BJB}), it follows that the matrix in (\ref{zzh}) is nonsingular:
\begin{equation}
\label{detmho}
    \det
    {\begin{bmatrix}
      I_\nu & \Theta\\
      A & -\Theta A^\rT
    \end{bmatrix}}
    =
    \det( - A\Theta - \Theta A^\rT) = \det \mho \ne 0  .
\end{equation}
The inversion of this matrix solves (\ref{zzh}) for $h$ as
\begin{equation}
\label{hzz}
    h
    =
    {\begin{bmatrix}
      I_\nu + \Theta \mho^{-1} A & -\Theta\mho^{-1}\\
      -\mho^{-1}A & \mho^{-1}
    \end{bmatrix}
    \begin{bmatrix}
      z\\
      \dot{z}
    \end{bmatrix}},
\end{equation}
thus allowing the boundary conditions (\ref{h+-0}) to be represented in terms of $z$, $\dot{z}$ by (\ref{z+0_z-T}).
Similarly to (\ref{gddot}), by differentiating (\ref{zdot}) and using the ODE (\ref{BVP}) together with the matrices (\ref{cABC}), the PR condition (\ref{PR}) and the relation (\ref{hzz}), it follows that
$
    \ddot{z}
     =
    \cC\cA \dot{h}
    =
    \cC\cA^2 h + \cC \cA \cB S^\rT f
    =
    {\small\begin{bmatrix}
      A^2 &  \Theta  (A^\rT)^2
    \end{bmatrix}
    \begin{bmatrix}
      I_\nu + \Theta \mho^{-1} A & -\Theta\mho^{-1}\\
      -\mho^{-1}A & \mho^{-1}
    \end{bmatrix}
    \begin{bmatrix}
      z\\
      \dot{z}
    \end{bmatrix}}
    +
    (A\Theta + \Theta  A^\rT) S^\rT f
    =
    \mho A^\rT \mho^{-1} A z  + (A-\mho A^\rT \mho^{-1}) \dot{z} -\mho S^\rT f
$,
which establishes (\ref{z''}). Here, use is also made of the identity
$
    A^2 \Theta - \Theta (A^\rT)^2 =
    A A\Theta - \Theta A^\rT A^\rT
    =
    -A(\mho + \Theta A^\rT) + (\mho + A\Theta )A^\rT
    =
    \mho A^\rT - A \mho
$, following from (\ref{PR}). Now, let (\ref{detS}) be satisfied. Then, in view of   (\ref{out}), the inclusion $f \in \ker \cL_T$ holds if and only if the function $z$ is identically zero over the time interval $[0,T]$ and hence,  so also are its derivatives,  including $\dot{z}$, $\ddot{z}$. The latter makes the left-hand side of the ODE (\ref{z''}) vanish identically, which  is equivalent to $f=0$ almost everywhere on $[0,T]$ since $\det (\mho S^\rT)\ne 0$ (both matrices $\mho$ and $S$ are nonsingular),  and hence, $\ker \cL_T = \{0\}$.
\end{proof}

The relation (\ref{detmho}) implies the observability\cite{AM_1971} of the matrix pair $(\cC, \cA)$ in (\ref{cABC}) and the controllability of the pair $(\cA, \cB)$. The latter property follows from the identity
$
    (\cC \cA^k)^\rT
    =
    (\cA^k)^\rT \cC^\rT
    =
    -
    (\cA^k)^\rT
    {\small\begin{bmatrix}
      0 & I_\nu\\
      I_\nu & 0
    \end{bmatrix}}
    \cB
    =
    -
    (-1)^k
    {\small\begin{bmatrix}
      0 & I_\nu\\
      I_\nu & 0
    \end{bmatrix}}
    \cA^k
    \cB
$ for all     $
    k = 0, 1, 2, \ldots$
and can also be obtained directly from $\det \mho\ne 0$ by using the Schur complement technique as
$
    \det
    {\small\begin{bmatrix}
      \cB & \cA \cB
    \end{bmatrix}}
    =
    \det
    {\small\begin{bmatrix}
      \Theta  & A\Theta \\
      -I_\nu & A^\rT
    \end{bmatrix}}
    =
    \det( A\Theta + \Theta A^\rT) = (-1)^\nu\det \mho \ne 0
$.
In the case of (\ref{detS}), the process $X$ in (\ref{X}) itself consists of the system variables of an OQHO with the CCR matrix $S\Theta S^\rT$ and the energy and coupling  matrices $S^{-\rT} RS^{-1}$, $MS^{-1}$ obtained by an appropriate transformation of those in (\ref{XCCR}), (\ref{AB}), with $S^{-\rT}:= (S^{-1})^\rT$. The corresponding transformation of the QSDE (\ref{dX}) is $    \rd X  = SAS^{-1}X  \rd t + SB\rd W
$.

Now, suppose the input bosonic fields are in the vacuum state\cite{P_1992} $\ups$ on the Fock space $\fF$, with the density operator $\rho$ in (\ref{bE}) given by
\begin{equation}
\label{rhoups}
    \rho
    :=
    \rho_0\ox \ups,
\end{equation}
where $\rho_0$ is the initial system state on $\fH_0$.
If, in addition to (\ref{rhoups}),
the matrix $A$ in (\ref{AB}) is Hurwitz,
then  the OQHO (\ref{dX}) has a unique invariant zero-mean Gaussian quantum state\cite{VPJ_2018a}. In this invariant  state, $X$ in (\ref{X}) is a stationary Gaussian quantum process satisfying (\ref{muPstat}), where, in addition to (\ref{Lambda}), which does not  depend on a particular state, the real part of the  quantum covariance function takes the form
\begin{equation}
\label{P}
    P(\tau)
     :=
    \left\{
    {\begin{matrix}
    S\re^{\tau A}\Gamma S^\rT& {\rm if}\  \tau\> 0\\
    S\Gamma\re^{-\tau A^{\rT}}S^\rT & {\rm if}\  \tau< 0\\
    \end{matrix}}
    \right.,
    \qquad
    \tau \in \mR.
\end{equation}
Here, $\Gamma= \int_{\mR_+} \re^{tA} BB^\rT \re^{tA^\rT} \rd t$ is the controllability Gramian\cite{AM_1971} of the matrix pair $(A,B)$ satisfying the algebraic Lyapunov equation (ALE)
$    A \Gamma + \Gamma A^\rT +  BB^\rT = 0$,
which can be combined with (\ref{PR}) as
$A (\Gamma+i\Theta) + (\Gamma+i\Theta) A^\rT +  B(I_m+iJ)B^\rT = 0$.
Since $A$ is Hurwitz, the quantum covariance function $P(\tau) + i\Lambda(\tau)$  of the process $X$ decays exponentially fast,  as $\tau\to \infty$, thus securing the integrability in (\ref{int}). The Fourier transforms (\ref{Phi0_Psi0}) of the
kernels (\ref{P}), (\ref{Lambda}) are rational functions (see also Eq.~(5.8) of Ref.~\refcite{VPJ_2019a}):
\begin{equation}
\label{Phi1_Psi1}
    \Phi(\lambda)
    =
    F(i\lambda) BB^\rT F(i\lambda)^*,
    \qquad
    \Psi(\lambda)
    =
    F(i\lambda) \mho F(i\lambda)^*,
    \qquad
    \lambda \in \mR,
\end{equation}
where use is made of the matrix $\mho$ from (\ref{PR}) along with the $\mC^{n\x \nu}$-valued transfer function
\begin{equation}
\label{F0}
    F(s)
    :=
    S
    (sI_\nu - A)^{-1},
    \qquad
    s \in \mC,
\end{equation}
which relates the Laplace transform of the process $X$ in (\ref{X}) to that of the incremented process $BW$ in (\ref{dX}):
$
    \int_{\mR_+}
    \re^{-st}
    X(t)
    \rd t
    =
    S
    (sI_\nu - A)^{-1}\cX(0)
    +
    F(s)B
    \int_{\mR_+}
    \re^{-st} \rd W(t)
$,
provided $\Re s >0$.  Since $A$ is Hurwitz,
the rational functions $\Phi$, $\Psi$ in (\ref{Phi1_Psi1}) have no poles in the strip $\{\lambda\in \mC: |\Im \lambda| < |\ln \br (\re^A)|\}$. The fulfillment of the sufficient  conditions of Theorem~\ref{th:non0} makes Theorems~\ref{th:limXi} and \ref{th:diff} applicable to computing the QEF growth rate (\ref{Ups}) for the invariant state of the stable OQHO in terms of the Fourier transforms (\ref{Phi1_Psi1}).
In this case, in view of (\ref{F0}),  the conditions (\ref{BJB}), (\ref{detS}) imply that $\det \Psi(\lambda)\ne 0$ for all $\lambda \in \mR$, which extends to the strip mentioned above.

\section{Quantum Statistical Uncertainty with an Entropy Theoretic Description}
\label{sec:worst}

In addition to the tail probability bounds (\ref{supP}), which apply to the  OQHO in its invariant Gaussian state, the QEF  growth rate $\Ups(\theta)$ provides asymptotic upper bounds
for the mean square cost functionals $\bE Q_T$ in the presence of statistical uncertainty described in terms of quantum relative entropy (see Sec.~IV of Ref.~\refcite{VPJ_2018b}
 and references therein).
More precisely, consider a setting,
 where the system-field density operator $\rho$ in (\ref{rhoups}) specifies a nominal quantum state which can differ  from the true density operator $\sigma$ on the system-field space $\fH$. Accordingly, $\bE(\cdot)$ in  (\ref{bE}) is interpreted as the nominal quantum expectation which, in general, differs  from its counterpart
\begin{equation}
\label{bEtrue}
    \bE_\sigma \zeta:= \Tr (\sigma \zeta)
\end{equation}
over the true state $\sigma$.
Since for any time horizon $T>0$,  the quantum variable $Q_T$ in (\ref{Q}) acts on the system-field subspace $\fH_T$ in (\ref{fHt}), its true and nominal moments depend on
\begin{equation}
\label{sigrhoT}
    \sigma_T:= \fP_T \sigma \fP_T,
    \qquad
    \rho_T:= \fP_T \rho \fP_T = \rho_0 \ox \ups_T,
\end{equation}
respectively, where $\ups_T$ is the vacuum field state on the Fock subspace $\fF_T$, and $\fP_T$ is the orthogonal projection of $\fH$ onto $\fH_T$ associated with the time interval $[0,T]$. In particular, the QEF $\Xi_{\theta,T}$ in (\ref{Xi}) is over the restricted nominal density operator $\rho_T$.  It is the true density operator $\sigma$, rather than its nominal model $\rho$,  that determines physically meaningful statistical properties of quantum variables (such as (\ref{bEtrue})), with the discrepancy between them constituting the quantum statistical uncertainty. If it is assumed a priori that $\sigma$, which is known imprecisely,  is not ``too far'' from the reference quantum state $\rho$,  this uncertainty can be described, for example,  by
the class
\begin{equation}
\label{fSeps}
    \fS_\eps
    :=
  \Big\{
    \sigma:\
    \limsup_{T\to +\infty}
    \Big(
    \frac{1}{T}
    \bD(\sigma_T \| \rho_T)
    \Big)
    \< \eps
  \Big\}.
\end{equation}
Here, the parameter $\eps\> 0$ limits the growth rate of the quantum relative entropy\cite{OW_2010} of the restricted true density operator $\sigma_T$ with respect to its nominal model $\rho_T$ in (\ref{sigrhoT}):
\begin{equation}
\label{bD}
  \bD(\sigma_T \| \rho_T)
  :=
  \bE_\sigma
  (\ln \sigma_T - \ln \rho_T)
  =
  \Tr
  (\sigma_T (\ln \sigma_T - \ln \rho_T))
  =
  - \bH(\sigma_T)
  -\bE_\sigma\ln \rho_T,
\end{equation}
where $\bH(\sigma_T)= -\bE_\sigma \ln \sigma_T = -\Tr(\sigma_T \ln \sigma_T)$ is the von Neumann entropy of $\sigma_T$ (cf. Eq.~(7) of Ref.~\refcite{YB_2009}), and the idempotence $\fP_T^2 = \fP_T$ of the projection operators is used. The supports of the density operators (as subspaces of $\fH_T$ in (\ref{fHt})) are assumed  to satisfy the inclusion $\supp \sigma_T \subset \supp \rho_T$ for all $T>0$,  which is a quantum analogue of the absolute continuity of classical probability measures.
Similarly to its classical counterpart\cite{CT_2006}, the quantum relative entropy (\ref{bD}) is always nonnegative and vanishes only if $\sigma_T  = \rho_T$. Therefore, the nominal state $\rho \in \fS_\eps$ and the parameter $\eps\> 0$ can be regarded as the ``centre'' and ``size'' of the uncertainty class $\fS_\eps$, respectively.
For large values of $\eps>0$, some states $\sigma \in \fS_\eps$ in (\ref{fSeps})  can lead to substantially higher values of the cost functionals than those predicted by the nominal model. In application to quadratic costs, the following theorem provides an upper bound on the worst-case growth rate for such a functional in terms of the nominal QEF.

\begin{theorem}
\label{th:worst}
Suppose the OQHO in (\ref{WWcomm})--(\ref{AB}) is stable (its matrix $A$ is Hurwitz), and the quantum variables $Q_T$ in (\ref{Q}) are associated with the process $X$ in (\ref{X}), with the operator $\cL_T$ in (\ref{cPT_cLT}) with   the CCR kernel $\Lambda$ in  (\ref{Lambda}) having no zero eigenvalues for all sufficiently  large $T>0$.
Also, suppose the class of true system-field states $\sigma$ is specified by  (\ref{fSeps}), with $\eps \> 0$.    Then the worst-case growth rate of the mean square cost $\bE_\sigma Q_T$ satisfies
\begin{equation}
\label{supEQ}
  \sup_{\sigma  \in \fS_\eps}
  \limsup_{T\to+\infty}
  \Big(
  \frac{1}{T}
  \bE_\sigma Q_T
  \Big)
  \<
  2
  \inf_{\theta >0\ {\rm subject\ to}\ (\ref{spec1})}
    \frac{\Ups(\theta) + \eps}{\theta}.
\end{equation}
Here, $\Ups(\theta)$ is the QEF growth rate (\ref{Ups}) computed in terms of the spectral functions (\ref{Phi1_Psi1}) for the invariant state of the OQHO driven by vacuum fields in the framework of the nominal model.
\hfill$\square$
\end{theorem}
\begin{proof}
Application of Lemma~2.1 from  Ref.~\refcite{YB_2009} and its corollary Eq.~(9) therein, based on the Golden-Thompson inequality\cite{G_1965,OP_1993,T_1965}
\begin{equation}
\label{GTI}
    \Tr(\re^{\alpha +\beta}) \< \Tr (\re^\alpha \re^\beta)
\end{equation}
for self-adjoint operators $\alpha$, $\beta$, to the quantum variables $\alpha := \ln \rho_T$ and $\beta := \frac{\theta}{2} Q_T$ yields
\begin{equation}
\label{GT}
  \frac{\theta}{2} \bE_\sigma Q_T
  \<
  \ln \Xi_{\theta,T}
  +
  \bD(
    \sigma_T
    \|
    \rho_T
  ),
\end{equation}
where the QEF $\Xi_{\theta,T} = \Tr (\rho_T \re^{\frac{\theta}{2}Q_T})$ in (\ref{Xi}) is over the nominal state $\rho$.  By combining (\ref{GT}) with (\ref{Ups}), (\ref{fSeps}), it follows that
\begin{equation}
\label{limsup}
  \frac{\theta}{2}
  \limsup_{T\to +\infty}
  \Big(
    \frac{1}{T}
    \bE_\sigma Q_T
  \Big)
  \<
  \limsup_{T\to +\infty}
  \Big(
    \frac{1}{T}
  \ln \Xi_{\theta,T}
  \Big)
  +
  \limsup_{T\to +\infty}
  \Big(
    \frac{1}{T}
  \bD(
    \sigma_T
    \|
    \rho_T
  )
  \Big)
  \<
    \Ups(\theta) + \eps
\end{equation}
for any $\theta>0$ satisfying (\ref{spec1}) and any $\sigma \in \fS_\eps$. The right-hand side of (\ref{limsup}), as an upper bound for the left-hand side, holds uniformly  over $\sigma \in \fS_\eps$, and hence, multiplication by $\frac{2}{\theta}$ leads to
$
    \sup_{\sigma \in \fS_\eps}
  \limsup_{T\to +\infty}
  \big(
    \frac{1}{T}
    \bE_\sigma Q_T
  \big)
  \<
    \frac{2}{\theta}(\Ups(\theta) + \eps)
$.
Since this inequality holds for any $\theta>0$
satisfying (\ref{spec1}),  and its left-hand side does not depend on $\theta$, the minimization of the right-hand side over $\theta$ yields (\ref{supEQ}).
\end{proof}

In the framework of the quantum statistical uncertainty model (\ref{fSeps}), specified by $\eps$ in terms of (\ref{bD}),  the worst-case quadratic cost growth rate on the left-hand side of (\ref{supEQ}) is similar to the robust performance criteria in classical minimax LQG control\cite{DJP_2000,P_2006,PJD_2000}. However, in contrast to the duality relation, which is used by the classical approach  in the context of a relative entropy description of statistical uncertainty, the equality in its quantum counterpart (\ref{GT}) is not necessarily achievable.

The uncertainty class (\ref{fSeps}) contains a smaller set of quantum states $\sigma$ on $\fH$,  given by
\begin{equation}
\label{fRepsc}
    \fR_{\eps, c}
    :=
  \Big\{
    \sigma:\
    \bD(\sigma_T \| \rho_T)
    \< \eps T + c(T)\
    {\rm for\ all}\ T>0
  \Big\}
  \subset \fS_\eps
\end{equation}
and parameterised by $\eps \> 0$ and an arbitrary function $c: (0,+\infty)\to \mR_+$ satisfying $c(T)=o(T)$, as $T\to +\infty$. In particular, if $\eps = 0$ and $c=0$,   then the set (\ref{fRepsc})  is a singleton consisting of the nominal state:
$\fR_{0,0} = \{\rho\}$, which corresponds to the absence of statistical uncertainty.
In general, a reasoning, similar to the proof of Theorem~\ref{th:worst}, leads to
\begin{equation}
\label{supEQ1}
  \sup_{\sigma  \in \fR_{\eps,c}}
  \limsup_{T\to+\infty}
  \Big(
  \frac{1}{T}
  \bE_\sigma Q_T
  \Big)
  \<
  \limsup_{T\to+\infty}
  \Big(
  \frac{1}{T}
  \sup_{\sigma  \in \fR_{\eps,c}}
  \bE_\sigma Q_T
  \Big)
   \<
  2
  \inf_{\theta >0\ {\rm subject\ to}\ (\ref{spec1})}
    \frac{\Ups(\theta) + \eps}{\theta}.
\end{equation}
Here, the maximization $  \sup_{\sigma  \in \fR_{\eps,c}}
  \bE_\sigma Q_T
$ of the quadratic cost $\bE_\sigma Q_T = \bE_{\sigma_T} Q_T$ under the quantum relative entropy constraint in (\ref{fRepsc})
is closely related to a quantum statistical mechanical problem\cite{S_1994}
\begin{equation}
\label{Helm1}
    -\frac{1}{2}\bE_\sigma Q_T + \frac{1}{\theta}\bD(\sigma_T \| \rho_T)
    =
    \bE_\sigma
    H_{\theta,T}
    -
    \frac{1}{\theta}
    \bH(\sigma_T)
    \longrightarrow
    \min
\end{equation}
of unconstrained minimization of a free energy functional over the density operators $\sigma_T$ on $\fH_T$. Here,
\begin{equation}
\label{H}
    H_{\theta,T}
    :=
            -
        \frac{1}{\theta} \ln \rho_T
    -\frac{1}{2}Q_T
\end{equation}
plays the role of a fictitious Hamiltonian which involves
both $Q_T$ and the nominal state $\rho_T$ on the system-field subspace $\fH_T$.
Accordingly, $\frac{1}{\theta}$ is the Lagrange multiplier for the constraint on $\bD(\sigma_T \| \rho_T)$, with $\theta$ corresponding to the inverse temperature (up to the Boltzmann constant).  The minimum in (\ref{Helm1}) is achieved at the Gibbs-Boltzmann density operator\cite{BB_2010}
$
    \frac{1}{Z_{\theta,T}} \re^{-\theta H_{\theta,T}}
$,
where, in view of (\ref{H}) and the Golden-Thompson inequality (\ref{GTI}), the partition function $Z_{\theta,T} :=     \Tr (\re^{-\theta H_{\theta,T}})
    =
    \Tr (\re^{\ln \rho_T + \frac{\theta}{2} Q_T })
$
is bounded by the nominal value of the QEF in (\ref{Xi}):
$
    Z_{\theta,T}
    \<
    \Tr (\rho_T \re^{\frac{\theta}{2}Q_T})
    =
    \Xi_{\theta,T}
$.

For any given $\eps\>0$, the quantity $\frac{1}{\theta}(\Ups(\theta)+\eps)$ under minimization  in (\ref{supEQ}) (and (\ref{supEQ1}))  is a convex function of $\theta>0$ due to each of the functions $\frac{\Ups(\theta)}{\theta}$ and $\frac{\eps}{\theta}$ being convex.
Also, since this quantity is increasing with respect to $\Ups(\theta)$, the  inequalities  (\ref{supEQ}), (\ref{supEQ1}) remain valid if the exact value of $\Ups(\theta)$ from (\ref{Ups}) is replaced with its upper bound. Such estimates for the QEF growth rate $\Ups(\theta)$ are provided, for example, by Theorem~5 of Ref.~\refcite{VPJ_2018a}.

In the context of risk-sensitive quantum feedback control and filtering problems\cite{B_1996,J_2004,J_2005,YB_2009}  with linear quantum plants and controllers or observers,  the resulting closed-loop systems are organised as OQHOs. The minimization of the QEF rate $\Ups(\theta)$ over admissible parameters of the controllers and observers (at a suitably chosen $\theta>0$),  as a performance criterion  for such systems, enhances their robustness properties in terms of the large deviations and worst-case quadratic cost bounds (\ref{supP}), (\ref{supEQ}), (\ref{supEQ1}).

\section{Conclusion}
\label{sec:conc}

For a quantum process of time-varying self-adjoint quantum variables with CCRs,  similar  to those of positions and momenta, we have developed a finite-horizon QKL expansion over the eigenbasis of the skew self-adjoint operator with the commutator kernel function, provided it has no zero eigenvalues. The QKL expansion has been applied to obtain a randomised representation of the QEFs for such quantum processes, which resembles the Doleans-Dade exponentials in the context of Girsanov's theorem and involves an auxiliary Gaussian random process whose covariance structure is specified by the commutator kernel of the underlying quantum process.
This representation has allowed the QEF to be related to the MGF of the quantum process and computed for the case of multipoint Gaussian states.
For stationary Gaussian quantum processes,  we have obtained a frequency-domain formula for the infinite-horizon asymptotic growth rate of the QEF in terms of the Fourier transforms of the real and imaginary parts of the quantum covariance function in composition with trigonometric functions. A homotopy technique has been outlined for numerical computation and approximation of the QEF growth rate as a function of the risk sensitivity parameter. A sufficient condition for no zero eigenvalues has been obtained in the case of  stationary Gaussian quantum processes, related linearly to the system variables of
stable OQHOs driven by vacuum fields, when the quantum covariance function has a rational Fourier transform and the eigenanalysis is reduced to a boundary value problem for a second-order ODE.
For this class of quantum systems, we have also discussed asymptotic upper bounds on the worst-case mean square costs   in terms of the QEF growth rate in the presence of statistical uncertainty described in terms of quantum relative entropy of the actual density operator with respect to the nominal state. In combination with exponential upper bounds on tail probabilities for quantum trajectories, the role of the QEF growth rate for the robustness properties of OQHOs makes it important to implement the frequency-domain representation, obtained in this paper,  in the form of state-space methods for its computation. The state-space approach would also benefit the solution of optimal control problems for OQHOs with quadratic-exponential performance criteria,  which are currently considered in the frequency domain\cite{VJP_2020}.

\section*{Acknowledgements}

IGV thanks Professor Robin L. Hudson for pointing out Ref.~\refcite{IT_2010} and for other useful discussions in the context of Refs.~\refcite{CH_2013,H_2018,VPJ_2019b}. This work is supported by the Australian Research Council grant DP210101938.

\appendix
\section{Randomised Representation for Elementary Quadratic-Exponential Functions}\label{sec:A}

For the purposes of Sec.~\ref{sec:quadro}, consider a self-adjoint quantum variable, which  is associated with the position and momentum operators $\xi$, $\eta$ in (\ref{zeta}), (\ref{posmomCCR}) of Sec.~\ref{sec:proc} as
\begin{equation}
\label{f}
  f(\sigma)
  :=
  \bM \re^{\sigma(\alpha \xi + \beta \eta )}
  =
  \frac{1}{2\pi}
  \int_{\mR^2}
    \re^{\sigma(a \xi + b\eta )
    -\frac{1}{2}(a^2 +b^2)}
  \rd a \rd b
\end{equation}
 and depends on a parameter $\sigma \in \mR$ satisfying the constraint
\begin{equation}
\label{sigmagood}
    |\sigma| < \sqrt{2}.
\end{equation}
The classical expectation $\bM(\cdot)$ in (\ref{f}) is over auxiliary independent standard normal  random variables $\alpha$, $\beta$.

\begin{theorem}
\label{th:fact}
The quantum variables (\ref{f}) are related to quadratic-exponential functions  of the position-momentum pair $(\xi, \eta)$ by
\begin{equation}
  \label{qefrand}
    \re^{\omega (\xi^2+\eta^2)}
    =
    \frac{1}{\cosh \omega}
    f(\sigma),
    \qquad
    \sigma := \sqrt{2\tanh \omega}
  \end{equation}
  for any $\omega\>0$.
\hfill$\square$
\end{theorem}
\begin{proof}
Differentiation of (\ref{f}) with respect to $\sigma$  yields
\begin{equation}
\label{f'}
    f'(\sigma)
     =
    \bM ((\alpha \xi + \beta \eta ) g(\sigma,\alpha,\beta))\\
    =
    \zeta^\rT
    \mu(\sigma),
\end{equation}
where
\begin{equation}
\label{BCH}
    g(\sigma, a, b)
    :=
    \re^{\sigma(a \xi + b \eta )}
    =
    \re^{-\frac{i}{2}\sigma^2 ab}
    \re^{\sigma a\xi}\re^{\sigma b\eta }
    =
    \re^{\frac{i}{2}\sigma^2 ab}
    \re^{\sigma b\eta }\re^{\sigma a\xi}
    =
    \re^{\sigma a(\xi-\frac{i}{2}\sigma b)}
    \re^{\sigma b\eta }
    =
    \re^{\sigma b(\eta+\frac{i}{2}\sigma a)}
    \re^{\sigma a\xi }
\end{equation}
is an auxiliary self-adjoint quantum variable which depends on the parameters $\sigma, a, b\in \mR$.  Here,
the position-momentum vector $\zeta$ from (\ref{zeta}) is used together with
\begin{equation}
\label{mu}
    \mu(\sigma)
    :=
      {\begin{bmatrix}
      \mu_1(\sigma)\\
      \mu_2(\sigma)
    \end{bmatrix}}
    =
     {\begin{bmatrix}
      \bM(\alpha g(\sigma,\alpha,\beta))\\
      \bM(\beta g(\sigma,\alpha,\beta))
    \end{bmatrix}},
\end{equation}
which is a vector consisting of two self-adjoint quantum variables, depending on the parameter $\sigma$.
The second and third equalities in (\ref{BCH}) follow from the Baker-Campbell-Hausdorff (BCH) formula (or, equivalently, the Weyl CCRs, similar to (\ref{Weyl})). In what follows, we will employ the identity
\begin{equation}
\label{nice}
    \bM (\gamma \re^{\gamma z}) = z\re^{\frac{1}{2}z^2} = z \bM \re^{\gamma z},
\end{equation}
which holds for a standard normal random variable $\gamma$ and extends from complex numbers to operators $z$. Since $\alpha$, $\beta$ are independent standard normal random variables, then substitution of the forth equality from (\ref{BCH}) into (\ref{mu}) and  application of the tower property of conditional classical expectations\cite{S_1996} together with (\ref{nice}) lead to
\begin{align}
\nonumber
    \mu_1(\sigma)
    & =
    \bM
    (
    \alpha
    \re^{\sigma\alpha (\xi-\frac{i}{2}\sigma \beta)}
    \re^{\sigma\beta \eta }
    )
    =
    \bM
    (
    \bM
    (
    \alpha
    \re^{\sigma\alpha (\xi-\frac{i}{2}\sigma \beta)}
    \mid
    \beta
    )
    \re^{\sigma\beta \eta }
    )\\
\nonumber
    &=
    \sigma
    \bM
    \Big(
    \bM
    \Big(
    \Big(\xi-\frac{i}{2}\sigma \beta\Big)
    \re^{\sigma\alpha
    (\xi-\frac{i}{2}\sigma \beta
    )}
    \,\Big|\,
    \beta
    \Big)
    \re^{\sigma\beta \eta }
    \Big)
    =
    \sigma
    \bM
    \Big(
    \Big(\xi-\frac{i}{2}\sigma \beta\Big)
    \re^{\sigma\alpha (\xi-\frac{i}{2}\sigma \beta)}
    \re^{\sigma\beta \eta }
    \Big)    \\
\label{Mg1}
    &=
    \sigma
    \bM
    \Big(
    \Big(\xi-\frac{i}{2}\sigma \beta\Big)
    g(\sigma,\alpha,\beta)
    \Big)
    =
    \sigma
    \xi
    f(\sigma)
    -
    \frac{i}{2}
    \sigma^2
    \mu_2(\sigma),
\end{align}
where use is also made of the relation $f(\sigma) = \bM g(\sigma, \alpha, \beta)$ in view of (\ref{f}). By  a similar reasoning, a combination of the last equality from (\ref{BCH}) with (\ref{mu}), (\ref{nice}) yields
\begin{align}
\nonumber
    \mu_2(\sigma)
    & =
    \bM
    (
    \beta
    \re^{\sigma\beta (\eta+\frac{i}{2}\sigma \alpha)}
    \re^{\sigma\alpha\xi }
    )
    =
    \bM
    (
    \bM
    (
    \beta
    \re^{\sigma\beta (\eta+\frac{i}{2}\sigma \alpha)}
    \mid
    \alpha
    )
    \re^{\sigma\alpha\xi }
    )\\
\nonumber
    &=
    \sigma
    \bM
    \Big(
    \bM
    \Big(
    \Big(\eta+\frac{i}{2}\sigma \alpha\Big)
    \re^{\sigma\beta (\eta+\frac{i}{2}\sigma \alpha)}
    \, \Big|\,
    \alpha
    \Big)
    \re^{\sigma\alpha\xi }
    \Big)
    =
    \sigma
    \bM
    \Big(
    \Big(\eta+\frac{i}{2}\sigma \alpha\Big)
    \re^{\sigma\beta (\eta+\frac{i}{2}\sigma \alpha)}
    \re^{\sigma\alpha\xi }
    \Big)    \\
\label{Mg2}
    &=
    \sigma
    \bM
    \Big(
    \Big(\eta+\frac{i}{2}\sigma \alpha\Big)
    g
    \Big)
    =
    \sigma
    \eta
    f(\sigma)
    +
    \frac{i}{2}
    \sigma^2
    \mu_1(\sigma).
\end{align}
The relations (\ref{Mg1}), (\ref{Mg2}) form a set of two linear equations for the vector $\mu(\sigma)$ in  (\ref{mu}), which can be represented by using the matrix (\ref{bJ}) and the vector (\ref{zeta}) as
$
    \mu(\sigma)
    =
    \sigma
    \zeta
    f(\sigma)
    -
    \frac{i}{2}
    \sigma^2
    \bJ
    \mu(\sigma)
$,
and hence,
\begin{equation}
\label{mu1}
    \mu(\sigma)
    =
    \sigma
    \Big(
    I_2
    +
    \frac{i}{2}
    \sigma^2
    \bJ
    \Big)^{-1}
    \zeta
    f(\sigma)
    =
    \frac{\sigma}{1-\frac{1}{4} \sigma^4}
    \Big(I_2
    -
    \frac{i}{2}
    \sigma^2
    \bJ
    \Big)
    \zeta
    f(\sigma)
\end{equation}
in view of the involutive property $(i\bJ)^2 = I_2$ (whereby $(I_2 + ic \bJ)(I_2 - ic \bJ) =
(1-c^2)I_2$ for any $c \in \mC$). Substitution of (\ref{mu1}) into (\ref{f'}) leads to
\begin{equation}
\label{f'1}
    f'(\sigma)
    =
    \frac{\sigma}{1-\frac{1}{4} \sigma^4}
    \zeta^\rT
    \Big(I_2
    -
    \frac{i}{2}
    \sigma^2
    \bJ
    \Big)
    \zeta
    f(\sigma)
    =
    \frac{\sigma}{1-\frac{1}{4} \sigma^4}
    \Big(
    \zeta^\rT \zeta
    +
    \frac{1}{2}
    \sigma^2
    \Big)
    f(\sigma),
\end{equation}
where use is made of the relations $\zeta^{\rT}\bJ \zeta = \xi\eta - \eta \xi = [\xi, \eta] = i$ which follow from (\ref{zeta})--(\ref{bJ}). Here, the denominator $1-\frac{1}{4} \sigma^4$ originates from the BCH correction factors $\re^{\pm \frac{i}{2}\sigma^2 \alpha \beta}$ in (\ref{BCH}), which explains the nature of the constraint (\ref{sigmagood}).
Now, (\ref{f'1}) is a linear operator differential equation for $f$ with the identity operator as the initial condition: $f(0)=\cI$ in view of (\ref{f}). Its solution is given by a leftward-ordered exponential
\begin{equation}
\label{fsol}
  f(\sigma)
  =
  \overleftarrow{\exp}
  \Big(
  \int_0^\sigma
    \frac{\tau}{1-\frac{1}{4} \tau^4}
    \Big(
    \zeta^\rT \zeta
    +
    \frac{1}{2}
    \tau^2
    \Big)
    \rd \tau
  \Big)
  =
  \exp
  \Big(
    \int_0^\sigma
    \frac{\tau}{1-\frac{1}{4} \tau^4}
    \rd \tau
    \zeta^\rT \zeta
    +
    \frac{1}{2}
    \int_0^\sigma
    \frac{\tau^3}{1-\frac{1}{4} \tau^4}
    \rd \tau
  \Big),
\end{equation}
which reduces to the usual exponential of an affine function of the quantum variable $\zeta^\rT \zeta$ (noncommutativity issues do not arise here because (\ref{fsol}) involves only one quantum variable). The coefficients of this affine function can be computed in terms of a new integration  variable $\omega\>0$, related to $\sigma$  by
\begin{equation}
\label{lamsig}
    \omega := \frac{1}{2}\ln\frac{1+\frac{1}{2}\sigma^2}{1-\frac{1}{2}\sigma^2},
    \qquad
    \sigma = \sqrt{2\frac{\re^{2\omega}-1}{\re^{2\omega}+1}} = \sqrt{2\tanh \omega},
\end{equation}
which describes a bijection $[0,+\infty)\ni \omega \leftrightarrow \sigma \in [0,\sqrt{2})$. More precisely, it follows from (\ref{lamsig}) that
$
    \frac{\sigma \rd \sigma}{1-\frac{1}{4} \sigma^4}
     =
    \frac{1}{2}
    \big(
        \frac{1}{1+\frac{1}{2} \sigma^2}
        +
        \frac{1}{1-\frac{1}{2} \sigma^2}
    \big)
    \rd (\sigma^2/2)
    =
    \rd \omega$
    and
    $\frac{1}{2}
    \frac{\sigma^3 \rd \sigma}{1-\frac{1}{4} \sigma^4}
     =
    -\frac{1}{2}
    \rd \ln
    \big(
        1-\frac{1}{4} \sigma^4
    \big)
    =
    -\frac{1}{2}
    \rd \ln
    (
        1-(\tanh \omega)^2
    )
    =
    \rd \ln \cosh \omega
$,
and hence, (\ref{fsol}) takes the form
\begin{equation}
\label{qeff}
    f(\sigma)
    =
    \re^{\omega \zeta^\rT \zeta + \ln \cosh \omega}
    =
    \re^{\omega \zeta^\rT \zeta} \cosh \omega.
\end{equation}
This can also be obtained by representing the ODE (\ref{f'1}) as $\rd f = f'\rd \sigma = (\zeta^\rT \zeta + \tanh \omega )f\rd \omega$ and using the relation $\int_0^\omega\tanh u \rd u = \ln \cosh \omega$.
It now remains to note that $\zeta^\rT \zeta = \xi^2 + \eta^2$ in view of (\ref{zeta}),  whereby (\ref{qeff}) establishes (\ref{qefrand}).
\end{proof}

\end{document}